\documentclass[pre,twocolumn,10pt,aps,longbibliography,nofootinbib]{revtex4-1}

 \usepackage{verbatim}
 \usepackage{amsmath}
 \usepackage{amssymb}
 \usepackage{amsthm}
 \usepackage{latexsym}
 \usepackage{amsfonts}
 \usepackage{epsfig}
 \usepackage{epstopdf}
 \usepackage{xcolor}
 \definecolor{darkblue}{rgb}{0,0,.5}
 \usepackage[linktocpage, colorlinks=true, linkcolor=darkblue, citecolor=darkblue]{hyperref}
 \usepackage[all]{hypcap}

\newcommand{\C}[1]{{\cal{#1}}}
\newcommand{\bb}[1]{\textbf{#1}}

\newcommand{\mf}[1]{{\mathfrak{#1}}}

\newcommand{\lr}[1]{{\left\langle {#1}\right\rangle}}
\newcommand{\rl}[0]{{\rangle\langle}}

\def\dbar{{\mathchar'26\mkern-12mu d}}

\begin{document}

\title{An operational approach to quantum stochastic thermodynamics}

\author{Philipp Strasberg}
\affiliation{Physics and Materials Science Research unit, University of Luxembourg, L-1511 Luxembourg, Luxembourg}
\affiliation{F\'isica Te\`orica: Informaci\'o i Fen\`omens Qu\`antics, Departament de F\'isica, Universitat Aut\`onoma de Barcelona, ES-08193 Bellaterra (Barcelona), Spain}

\date{\today}

\begin{abstract}
 We set up a framework for quantum stochastic thermodynamics based solely on experimentally controllable, but 
 otherwise arbitrary interventions at discrete times. Using standard assumptions about the system-bath dynamics and 
 insights from the repeated interaction framework, we define internal energy, heat, work and entropy at the trajectory 
 level. The validity of the first law (at the trajectory level) and the second law (on average) is established. The 
 theory naturally allows to treat incomplete information and it is able to smoothly interpolate between a 
 trajectory based and ensemble level description. We use our theory to compute the thermodynamic efficiency of recent 
 experiments reporting on the stabilization of photon number states using real-time quantum feedback control. Special 
 attention is also payed to limiting cases of our general theory, where we recover or contrast it with previous 
 results. We point out various interesting problems, which the theory is able to address rigorously, such as the 
 detection of quantum effects in thermodynamics. 
\end{abstract}

\maketitle

\newtheorem{mydef}{Definition}[section]
\newtheorem{lemma}{Lemma}[section]
\newtheorem{thm}{Theorem}[section]
\newtheorem{crllr}{Corollary}[section]
\theoremstyle{remark}
\newtheorem{rmrk}{Remark}[section]

\section{Introduction}

The nonequilibrium thermodynamics of small Markovian systems is well-studied for decades if we are interested 
only in ensemble averaged quantities of internal energy, heat, work or entropy~\cite{SchnakenbergRMP1976, HillBook1977, 
SpohnLebowitzAdvChemPhys1979, AlickiJPA1979, LindbladBook1983, KosloffEntropy2013}. 
For classical systems it became clear during the past 25 years that also fluctuations in thermodynamic quantities bear 
important information and that those fluctuations are constrained by fundamental symmetry relations valid arbitrary far 
from equilibrium. These symmetry relations are known as fluctuation theorems~\cite{EvansSearlesAdvPhy2002, 
JarzynskiAnnuRevCondMat2011}. For a given realization of a stochastic process an understanding of the 
fluctuation theorem required to extend the ensemble averaged energetic~\cite{SekimotoPTPS1998, SekimotoBook2010} and 
entropic~\cite{SeifertPRL2005} description to the level of single stochastic trajectories. The resulting theoretical 
framework is called stochastic thermodynamics~\cite{SeifertRPP2012, VandenBroeckEspositoPhysA2015}. 

Quantum stochastic thermodynamics tries to generalize classical stochastic thermodynamics to systems whose quantum 
nature cannot be neglected. Obviously, the very definition of a trajectory dependent quantity is non-trivial as any 
measurement disturbs the system and the meaning of a `trajectory' is \emph{a priori} not clear. We note that 
incomplete and disturbing measurements are also prevalent in classical systems~\cite{BechhoeferRMP2005}, but exploring 
their consequences for classical stochastic thermodynamics has raised relatively little attention so 
far~\cite{RibezziCrivellariRitortPNAS2014, AlemanyRibezziCrivellariRitortNJP2015, BechhoeferNJP2015, 
GarciaGarciaLahiriLacostePRE2016, WaechtlerStrasbergBrandesNJP2016, PolettiniEspositoPRL2017, 
PolettiniEspositoJSP2019}. 

Soon after the discovery of classical fluctuation theorems, much effort was devoted to derive fluctuation theorems for 
quantum systems. A theoretically successful strategy is the two-point measurement 
approach~\cite{EspositoHarbolaMukamelRMP2009, CampisiHaenggiTalknerRMP2011}. It requires to measure the 
energy of the system \emph{and} the bath at the beginning and at the end of the thermodynamic 
process. Obviously, for a bath with its prosaic $10^{23}$ degrees of freedom such a scheme is not even for a classical 
system practically feasible. In addition, the resulting statistics for internal energy and work cannot fulfill the 
first law if the initial state is not diagonal in the energy eigenbasis~\cite{PerarnauLlobetEtAlPRL2017}. 
Nevertheless, within this approach quantum fluctuation theorems can be derived, which are formally identical to their 
classical counterpart. Thus, by measuring the whole universe (system plus bath), the two-point measurement approach 
circumvents the need to define thermodynamic quantities along a specific system trajectory. Also alternative and 
complementary approaches based on interferometric measurements~\cite{MazzolaDeChiaraPaternostroPRL2013, 
DornerEtAlPRL2013, BatalhaoEtAlPRL2014, SolinasGasparinettiPRE2015, SolinasGasparinettiPRA2016}, a single projective 
measurement~\cite{CerrilloBuserBrandesPRB2016, CerisolaEtAlNatComm2017} or no measurement at all~\cite{AbergPRX2018, 
WhitneyPRB2018} have been put forward and the semiclassical limit was studied too~\cite{JarzynskiQuanRahavPRX2015, 
ZhuEtAlPRE2016, GarciaMataEtAlPRE2017}. To conclude, even though those approaches are theoretically powerful, they are 
experimentally hard to confirm and an important feature of classical stochastic thermodynamics is still missing, namely 
the definition of internal energy and entropy along a given `quantum trajectory'. 

Exceptions are quantum systems which, when perfectly observed in the energy eigenbasis, follow a Markovian rate master 
equation. This is approximately the case in electronic nanostructures (quantum dots) in the sequential tunneling 
regime~\cite{UtsumiEtAlPRB2010, KungEtAlPRX2012, SairaEtAlPRL2012, SchallerBook2014}, where the framework of classical 
stochastic thermodynamics was carried over one by one. Interestingly, trying to adopt this picture to more general 
quantum dynamics results in unconventional definitions for thermodynamic quantities~\cite{EspositoMukamelPRE2006}, not 
to mention the measurement problem. This further demonstrates the need for a radically different approach to quantum 
stochastic thermodynamics. 

One such approach makes use of the framework of repeated interactions~\cite{HorowitzPRE2012, HorowitzParrondoNJP2013, 
BenoistEtAlCM2018}. In there, the static bath is replaced by an external stream of ancilla systems, which are put into 
contact with the system one by one and are designed to simulate a thermal bath (arbitrary initial states of the bath 
were recently treated in Ref.~\cite{ManzanoHorowitzParrondoPRX2018}). If the external systems are projectively measured 
before and after the interaction, a trajectory based formulation becomes possible similar to classical stochastic 
thermodynamics. Although such a description yields theoretical insights, in experimental reality a system is usually 
also in permanent contact with a bath. 

An experimentally closer approach uses a technique, which was discovered in quantum optics in order to describe the 
stochastic evolution of a quantum system based on monitoring the environment of the 
system~\cite{DalibardCastinMolmerPRL1992, GardinerParkinsZollerPRA1992, CarmichaelBook1993}. Given such a measurement 
scheme, the system dynamics can be `unraveled' by describing it in terms of a stochastic Schr\"odinger or master 
equation. Combined with this dynamical description, researchers recently applied the ideas of stochastic thermodynamics 
to such quantum systems~\cite{HekkingPekolaPRL2013, AlonsoLutzRomitoPRL2016, ElouardEtAlQInf2017, DresselEtAlPRL2017, 
ElouardEtAlNJP2017, ManikandanElouardJordanPRA2019, ElouardMohammadyBook2018}; a completely general picture is, however, 
still missing. For instance, a trajectory dependent system entropy was never introduced making it hard to study entropy 
production along a single trajectory or on average (specific fluctuation theorems based on a particular choice of the 
backward dynamics were studied in Refs.~\cite{ElouardEtAlQInf2017,  ElouardEtAlNJP2017, ManikandanElouardJordanPRA2019, 
ElouardMohammadyBook2018}; we will come back to this at the end). Furthermore, the above publications focused only on 
efficient measurements in which the state of the system along a particular trajectory is always pure (for some specific 
scenarios first steps were already undertaken to overcome this limitation~\cite{AlonsoLutzRomitoPRL2016, 
ElouardEtAlNJP2017}). Finally, only simple protocols excluding feedback control have been studied so far 
(Refs.~\cite{AlonsoLutzRomitoPRL2016, ElouardEtAlQInf2017} consider also very simple 
feedback schemes for specific systems). 

To conclude, apart from a few model specific studies, a common feature of \emph{all} previous approaches 
is the reliance on a perfectly monitored system and environment such that the system is always in a pure state 
along every trajectory. In this sense, there is no essential departure from the two-point measurement scheme in which 
perfect knowledge of every involved degree of freedom is crucial. 

\subsection{Results and outline}

We here put forward a novel approach, which we propose to call \emph{operational} quantum stochastic thermodynamics 
because it places the experimenter in the foreground. A `stochastic trajectory' -- and the corresponding thermodynamic 
quantities internal energy, heat, work and entropy along such a trajectory -- are defined \emph{solely} in terms of 
experimentally meaningful interventions or control operations of the system dynamics. Dynamically, our description 
rests on recent theoretical progress in describing `quantum causal models' or `quantum stochastic 
processes'~\cite{ChiribellaDArianoPerinottiPRL2008, ChiribellaDArianoPerinottiPRA2009, CostaShrapnelNJP2016, 
OreshkovGiarmatziNJP2016, AllenEtAlPRX2017, PollockEtAlPRL2018,  PollockEtAlPRA2018, MilzPollockModiPRA2018, 
MilzEtAlArXiv2017, SakuldeeEtAlJPA2018}. Within this picture it is possible to describe the effect of arbitrary 
control operations happening at arbitrary discrete times applied to an arbitrary quantum system in an experimentally 
measurable way. It is different from conventional quantum trajectory approaches and we will start the paper by 
discussing it in Sec.~\ref{sec process tensor}. 

In Sec.~\ref{sec process tensor repeated interactions} we then connect this approach to the framework of repeated 
interactions. Partially based on insights from earlier work~\cite{StrasbergEtAlPRX2017}, we will see in 
Sec.~\ref{sec general stoch thermo} that this allows us to find an unambiguous first and second law of 
thermodynamics for each single control operation. 

The only standard assumption we are here using is that the system in \emph{absence} of control operations can be modeled 
by a quantum master equation with a transparent thermodynamic interpretation describing a driven system coupled to a 
single heat bath.\footnote{An extension beyond this 
Markovian picture is, however, possible in some cases, see Sec.~\ref{sec outlook}.} Based on the repeated interaction 
picture, we will then see in Sec.~\ref{sec general stoch thermo} that the definitions of internal energy and system 
entropy emerge naturally out of the framework if we properly take into account all interacting subsystems. In fact, 
following the credo ``information is physical''~\cite{LandauerPhysTod1991}, we will see that it is necessary to include 
the full information generated by the measurements into the entropic balance from the beginning on. With this step 
we also depart from the approaches reviewed above, which need to be modified in presence of feedback control (see 
Ref.~\cite{ParrondoHorowitzSagawaNatPhys2015} for an introduction). The first law at the trajectory level 
and the second law on average is finally verified. 

This concludes the first part of the manuscript, which is about the basic framework of operational quantum stochastic 
thermodynamics. Its novelties are: 

(1) It does neither rely on the ability to have control about the environment nor does it require continuous 
measurements. 

(2) By allowing to treat any kind of incomplete information, it respects experimental reality where every measurement is imprecise and imperfect. 

(3) It shows that any conceivable feedback scenario has a consistent thermodynamic interpretation.\footnote{This 
includes the case of real-time feedback control, where -- in contrast to deterministic feedback control where the time 
of measurement and feedback are pre-determined~\cite{ParrondoHorowitzSagawaNatPhys2015} -- the control strategy is 
adapted during the run of the experiment. It also includes the case of time-delayed feedback control. }

(4) The notion of stochastic entropy for a quantum system is defined and the second law follows without the need to 
introduce any `backward' dynamics. 

(5) The framework reveals that quantum stochastic thermodynamics is \emph{more} than a mere extension of classical 
stochastic thermodynamics. Any measurement strategy has in general a non-trivial impact on the quantum system and hence, 
there is a plurality of first and second laws in quantum thermodynamics depending on how we measure the system. Notice 
that these many laws of thermodynamics are conceptually different from the many second laws of 
Ref.~\cite{BrandaoEtAlPNAS2014}. 

The rest of the paper is about illuminating applications and special cases of the general theory: 

(6) To illustrate point (2) and (3), we analyze in Sec.~\ref{sec example} the quantum stochastic thermodynamics 
of recent experiments reporting on the preparation and stabilization of photon number states~\cite{SayrinEtAlNature2011, 
ZhouEtAlPRL2012}. We uncover that the efficiency to \emph{prepare} such states is remarkably high. 

(7) We consider the case of projective measurements in detail and compare our definitions with the recently 
introduced notion of ``quantum heat''~\cite{ElouardEtAlQInf2017} in Sec.~\ref{sec projective measurements}. 

(8) In Sec.~\ref{sec two point meas approach} we provide a resolution to the no-go theorem derived by Perarnau-Llobet 
\emph{et al.}~\cite{PerarnauLlobetEtAlPRL2017}, which (in a nutshell) shows that the conventional definition of work 
used in the two-point measurement scheme~\cite{EspositoHarbolaMukamelRMP2009, CampisiHaenggiTalknerRMP2011} is 
doubtful. Indeed, we show that it is inconsistent with our definition of stochastic work. 

(9) Secs.~\ref{sec standard quantum thermo},~\ref{sec repeated interactions ensemble level} 
and~\ref{sec standard stoch thermo} provide important consistency checks. We show that the definitions of standard 
quantum thermodynamics~\cite{SpohnLebowitzAdvChemPhys1979, AlickiJPA1979, LindbladBook1983, KosloffEntropy2013} and 
the repeated interaction framework~\cite{StrasbergEtAlPRX2017} are contained in our general approach. They arise, 
however, \emph{not} by averaging over many trajectories, but by deciding not to do any measurements at all. 
In the limit of a perfectly observed classical system we recover the definitions of internal energy, heat and work 
of standard stochastic thermodynamics. Only our second laws differ because our framework remains valid in case of 
feedback control, whereas the conventional framework~\cite{SekimotoBook2010, SeifertRPP2012, 
VandenBroeckEspositoPhysA2015} needs to be modified then~\cite{ParrondoHorowitzSagawaNatPhys2015}. 

(10) In Secs.\ref{sec getting rid of units} and~\ref{sec quantum stoch thermo without theory} we discuss particularly 
interesting cases, which allow to reduce the complexity of our general framework. 

The paper ends with some remarks and an outlook. Sec.~\ref{sec final remarks} discusses the case of multiple heat 
baths, possible `second laws' that follow from a time-reversed process, and the necessity to use the repeated 
interaction framework and to focus on incomplete information from the beginning on. In Sec.~\ref{sec outlook} we point 
out to interesting future applications such as finding true quantum features in quantum heat engines, relations to 
Leggett-Garg inequalities and the detection of non-Markovian effects in thermal machines. 

\subsection{Basic notation}

The state of a system $X$ at time $t$ is described by a density operator $\rho_X(t)$. The corresponding Hilbert space 
of the system is denoted by $\C H_X$ and the Hamiltonian by $H_X$ or $H_X(\lambda_t)$ if it depends on an externally 
controlled time-dependent parameter $\lambda_t$. The von Neumann entropy of an arbitary state $\rho_X$ is defined as 
$S_\text{vN}(\rho_X) \equiv -\mbox{tr}_X\{\rho_X\ln\rho_X\}$ and the Shannon entropy of an arbitrary probability 
distribution $p(x)$ is $S_\text{Sh}[p(x)] \equiv -\sum_x p(x)\ln p(x)$. To characterize the correlations of a 
bipartite system $XY$ in state $\rho_{XY}$, we use the always positive mutual information 
$I_{X:Y} \equiv S_\text{vN}(\rho_X) + S_\text{vN}(\rho_Y) - S_\text{vN}(\rho_{XY})$. It is closely related to the 
always positive relative entropy $D[\rho||\sigma]\equiv\mbox{tr}\{\rho(\ln\rho-\ln\sigma)\}$ by noting that 
$I_{X:Y} = D[\rho_{XY}\|\rho_X\otimes\rho_Y]$ where $\rho_{X/Y} \equiv \mbox{tr}_{Y/X}\{\rho_{XY}\}$ 
denotes the marginal state. Furthermore, we denote superoperators, which map operators onto operators, by calligraphic 
letters, e.g., $\C U, \C V, \C P$, etc.

Below, we will see that a stochastic trajectory is specified by a sequence of measurement results or outcomes 
$r_n,\dots,r_1$, which were obtained at times $t_n> \dots> t_1$. The sequence of outcomes will be denoted by 
$\bb r_n \equiv (r_n,\dots,r_1)$. The state of a system $X$ at time $t>t_n$ conditioned on such a sequence will be 
denoted by $\rho_X(t,\bb r_n)$. The ensemble averaged state is given by 
$\rho_X(t) = \sum_{\bb r_n} p(\bb r_n) \rho_X(t,\bb r_n)$ where $p(\bb r_n)$ denotes the probability of obtaining 
the sequence of outcomes $\bb r_n$. We will also keep this notation for thermodynamic quantities such as internal 
energy $E$, heat $Q$, work $W$ and entropy $S$ (which possibly have additional sub- and superscripts). This means, for 
instance, that the stochastic internal energy depending on the outcomes $\bb r_n$ is denoted by $E(t,\bb r_n)$ 
whereas the ensemble averaged internal energy is written $E(t) = \sum_{\bb r_n} p(\bb r_n) E(t,\bb r_n)$. 

\section{The process tensor}
\label{sec process tensor}

Classical stochastic thermodynamics is based on the theory of classical stochastic processes. A corresponding 
quantum thermodynamic framework needs to be based on the theory of quantum stochastic processes. There has been 
recently large progress on this topic and we will here use the process tensor to represent a quantum stochastic 
process~\cite{PollockEtAlPRL2018, PollockEtAlPRA2018, MilzPollockModiPRA2018, MilzEtAlArXiv2017, SakuldeeEtAlJPA2018}.
It is the extension of `quantum superchannels'~\cite{ChiribellaDArianoPerinottiEPL2008, ModiSR2012} to multiple control 
operations and it is closely related to the `quantum comb' framework studied in 
Refs.~\cite{ChiribellaDArianoPerinottiPRL2008, ChiribellaDArianoPerinottiPRA2009}. Similar frameworks have been 
also developed within the emergent field of quantum causal modelling~\cite{CostaShrapnelNJP2016, 
OreshkovGiarmatziNJP2016, AllenEtAlPRX2017} and even earlier attempts in that direction can be found in 
Refs.~\cite{LindbladCMP1979, AccardiFrigerioLewis1982}. The basic insight behind this formulation is to treat 
the control operations performed on the system as the elementary objects and not the state of the system itself 
because the latter can in general not be fully controlled. Here, the terminology `control operation' is used in a wide 
sense and could describe any action of an external agent such as measurements, unitary kicks, state preparations, noise 
addition, feedback control operations, etc. Mathematically, we only require that each control operation is described by 
a completely positive (CP) map. The following review about the basics of the process tensor requires some knowledge 
about quantum operations and quantum measurement theory, see Refs.~\cite{KrausBook1983, NielsenChuangBook2000, 
HolevoBook2001b, WisemanMilburnBook2010, JacobsBook2014} for introductory texts. 

As usual we consider a system $S$ coupled to a bath $B$ described by an arbitrary initial system-bath state 
$\rho_{SB}(t_0)$. The composite system-bath state evolves unitarily up to time $t_1\ge t_0$ according to the 
Liouville-von Neumann equation $\partial_t\rho_{SB}(t) = -i[H_\text{tot}(\lambda_t),\rho_{SB}(t)]$ ($\hbar\equiv1$) 
with global Hamiltonian 
\begin{equation}
 H_\text{tot}(\lambda_t) = H_S(\lambda_t) + H_{SB} + H_B.
\end{equation}
Here, the system Hamiltonian $H_S$ might depend on some arbitrary time dependent control protocol $\lambda_t$, but not 
the interaction Hamiltonian $H_{SB}$ and the bath Hamiltonian $H_B$. 
The resulting unitary evolution is described by the superoperator 
\begin{equation}\label{eq superoperator unitary}
 \C U_{1,0}\rho_{SB}(t_0) \equiv U(t_1,t_0)\rho_{SB}(t_0) U^\dagger(t_1,t_0)
\end{equation}
where $U(t_1,t_0) \equiv \C T_+\exp[-i\int_{t_0}^{t_1} dt H_\text{tot}(\lambda_t)]$ with the time ordering operator 
$\C T_+$. 

Then, at time $t_1 > t_0$ we interrupt the evolution by a CP operation $\C A(r_1)$, which only acts on the system and 
yields `outcome' $r_1$ (for instance, the result of a projective measurement). Mathematically, we write the operation as 
\begin{equation}\label{eq step 1}
 \tilde\rho_{SB}(t_1^+,r_1) = [\C A(r_1)\otimes\C I_B]\rho_{SB}(t_1^-).
\end{equation}
Here, $t_1^\pm = \lim_{\epsilon\searrow0} (t_1\pm\epsilon)$ denotes a time shortly after or before $t_1$ and $\C I_B$ 
denotes the identity superoperator acting on $B$. Note that we assume the control operation to happen instantaneously. 
It ensures that the experimenter has complete control 
over the operation: if the control operations takes longer, it would also affect the bath and a clear separation of the 
dynamics into a dynamics induced by the bath or the external agent becomes problematic. The final state of 
knowlegde after the operation $\tilde\rho_{SB}(t_1^+,r_1)$ can explicitly depend on the outcome $r_1$. Since $\C A(r_1)$ 
is CP, it admits an operator-sum (Kraus) representation of the form 
\begin{equation}\label{eq CP map}
 \C A(r_1)\rho_S = \sum_\alpha A_\alpha(r_1)\rho_SA_\alpha^\dagger(r_1),
\end{equation}
but we do not require it to be trace perserving (TP). For this reason we have used a `tilde' in Eq.~(\ref{eq step 1}) 
to emphasize that the state is not normalized. The probability to observe outcome $r_1$ at time $t_1$ is 
$p(r_1) = \mbox{tr}_{SB}\{\tilde\rho_{SB}(t_1^+,r_1)\}$. Then, the normalized system state after the control operation 
at time $t_1$ becomes $\rho_{S}(t_1^+,r_1) = \C A(r_1)\rho_S(t_1^-)/p(r_1)$. Notice that the map $\C A(r_1)/p(r_1)$ is 
CPTP, but non-linear in the state $\rho_S(t_1^-)$. It is the quantum analog of Bayes' rule. The average system state is 
accordingly 
\begin{equation}
 \rho_S(t_1^+) = \sum_{r_1} p(r_1)\rho_{S}(t_1^+,r_1) = \sum_{r_1}\C A(r_1)\rho_S(t_1^-).
\end{equation}
This would also correspond to our state of knowledge if we ignore the outcome $r_1$. Notice that the average 
control operation $\sum_{r_1} \C A(r_1)$ is now a CPTP map and can be written as 
\begin{equation}
 \sum_{r_1} \C A(r_1)\rho_S = \sum_{r_1,\alpha} A_\alpha(r_1)\rho_S A_\alpha^\dagger(r_1)
\end{equation}
with $\sum_{r_1,\alpha} A_\alpha^\dagger(r_1) A_\alpha(r_1) = 1_S$. 

We then iterate the above procedure by letting the joint system-bath state evolve unitarily up to time $t_2\ge t_1$: 
$\rho_{SB}(t_2^-,r_1) = \C U_{2,1}(r_1)\rho_{SB}(t_1^+,r_1)$. Now, however, the unitary operation is allowed to depend 
on $r_1$ by changing the control protocol of the system Hamiltonian $H_S[\lambda_t(r_1)]$. This actually corresponds to 
the simplest form of measurement-based quantum feedback control. Then, at time $t_2$ we subject the system to another 
CP control operation $\C A(r_2|r_1)$, which is also allowed to depend on $r_1$ and which gives outcome $r_2$. Thus, 
$\rho_{SB}(t_2^+,\bb r_2) = [\C A(r_2|r_1)\otimes\C I_B]\rho_{SB}(t_2^-,r_1)$, where $\bb r_2 = (r_2,r_1)$. 

We can re-iterate the above procedure by letting the external agent interrupt the unitary system-bath 
evolution at times $t_n> t_{n-1}> \dots> t_1$. Let us denote by $t$ an arbitrary time after the $n$'th but before 
the $(n+1)$'th control operation, i.e., $t_{n+1} > t > t_n$. The unnormalized state of the system conditioned on the 
sequence of outcomes $\bb r_n$ at such a time $t$ is then given by 
\begin{align}
  & \tilde\rho_{S}(t,\bb r_n) = \mf T[\C A(r_n|\bb r_{n-1}),\dots,\C A(r_1)]	\\
  & \equiv \mbox{tr}_B\left\{\C U_{t,n}(\bb r_n)\C A(r_n|\bb r_{n-1})\dots\C U_{2,1}(r_1)\C A(r_1)\C U_{1,0}\rho_{SB}(t_0)\right\}. \nonumber
\end{align}
Here, we have introduced the \emph{process tensor} $\mf T$. Its variable inputs are the set of control operations 
$\{\C A(r_i|\bb r_{i-1})\}_{i=1}^n$, but \emph{not} the initial state of the system, the bath or the composite. 
The trace of the process tensor gives the probability to observe the sequence of outcomes $\bb r_n$, 
\begin{equation}\label{eq probability process tensor}
 p(\bb r_n) = \mbox{tr}_S\{\mf T[\C A(r_n|\bb r_{n-1}),\dots,\C A(r_1)]\}
\end{equation}
such that the normalized state of the system can be written as 
\begin{equation}
 \rho_{S}(t,\bb r_n) = \frac{\mf T[\C A(r_n|\bb r_{n-1}),\dots,\C A(r_1)]}{p(\bb r_n)}.
\end{equation}
The process tensor is an operationally well-defined object for any open system dynamics (in particular for any 
environment) for any possible, physically admissible form of interventions in an experiment. It is different from 
typical quantum trajectory methods or quantum jump expansions~\cite{DalibardCastinMolmerPRL1992, 
GardinerParkinsZollerPRA1992, CarmichaelBook1993, HolevoBook2001b, WisemanMilburnBook2010, JacobsBook2014}, which rely 
on \emph{continuously} monitoring the \emph{environment} of the system. This framework is included as a limiting case 
in the process tensor, but it does not rely on it: any set of discrete times is allowed and the (often uncontrollable) environment does not need to be monitored. For further research on this topic see 
Refs.~\cite{ChiribellaDArianoPerinottiPRL2008, ChiribellaDArianoPerinottiPRA2009, CostaShrapnelNJP2016, 
OreshkovGiarmatziNJP2016, AllenEtAlPRX2017, PollockEtAlPRL2018,  PollockEtAlPRA2018, MilzPollockModiPRA2018, 
MilzEtAlArXiv2017, SakuldeeEtAlJPA2018}.

\section{Process tensor from repeated interactions}
\label{sec process tensor repeated interactions}

In practise the control operations $\C A(r_n|\bb r_{n-1})$ do not happen spontaneously, but require an active 
intervention from the outside. They are typically implemented by letting the system interact for a short 
time with an externally prepared apparatus (e.g., a memory or detector). It is the interaction time and the initial 
state of the apparatus, which can be usually well-controlled experimentally. This insight will naturally 
lead us to the framework of repeated interactions, in which we will model at least parts of the external apparatus 
explicitly. 

The main mathematical insight of this section rests on Stinespring's theorem~\cite{StinespringPAMS1955}, which states 
that any CPTP map $\C A$ can be seen as the reduced dynamics of some unitary evolution in an extended space. More 
precisely, we can always write 
\begin{equation}\label{eq Stinespring CPTP}
 \C A\rho_S = \mbox{tr}_U\{V \rho_S\otimes\rho_U V^\dagger\},
\end{equation}
where we labeled the additional subsystem by $U$ for `unit' in view of the thermodynamic framework considered later on 
and in unison with Ref.~\cite{StrasbergEtAlPRX2017}. The unit is in an initial state $\rho_U$ and $V$ denotes the 
unitary operator which acts jointly on $SU$. Furthermore, any non-trace preserving CP map $\C A(r)$ with outcome 
$r$ can be modeled as~\cite{HolevoBook2001b} 
\begin{equation}\label{eq Stinespring CP}
 \C A(r)\rho_S = \mbox{tr}_U\{P_U(r)V \rho_S\otimes\rho_U V^\dagger P_U(r)\},
\end{equation}
where each positive operator $P_U(r)$ acts only on $\C H_U$ and fulfills $\sum_r P^2_U(r) = 1_U$. Notice that 
Eq.~(\ref{eq Stinespring CPTP}) can be recovered from Eq.~(\ref{eq Stinespring CP}) either by choosing $P_U(r) = 1_U$ 
or by summing over $r$. In accordance with our previous superoperator notation, we introduce 
$\C P_U(r)\rho_U \equiv P_U(r)\rho_U P_U(r)$ and $\C V\rho_{SU} \equiv V\rho_{SU} V^\dagger$ such that we can write 
Eq.~(\ref{eq Stinespring CP}) in the shorter form $\C A(r)\rho_S = \mbox{tr}_U\{\C P_U(r)\C V \rho_S\otimes\rho_U\}$. 

It is worth to remark that the above representation of the control operation is not unique. What we are aiming at 
here is a \emph{minimal} consistent thermodynamic description for any given set of control operations. If additional 
physical insights are available, they have to be taken into account (see Sec.~\ref{sec example} for a clear 
experimental example). The only important point, however, is that the general operator-sum 
representation~(\ref{eq CP map}) can be decomposed into more primitive operations (a unitary and a measurement of the 
unit). 

The whole process tensor $\mf T[\C A(r_n|\bb r_{n-1}),\dots,\C A(r_1)]$ can then be seen as describing the reduced 
dynamics of a system coupled to a stream of units, which interact sequentially at times $t_n > \dots > t_1$ with the 
system, see Fig.~\ref{fig setup}. This constitutes the \emph{framework of repeated interactions}. Then, the unnormalized 
joint state of the system and all units, which have interacted with the system up to time $t$ ($t_{n+1} > t > t_n$) with 
outcome $\bb r_n$, can be written as 
\begin{widetext}
 \begin{align}
  & \tilde\rho_{SU(\bb n)}(t,\bb r_n) = \label{eq state SUn general}    \\
  & \mbox{tr}_B\big\{\C U_{t,t_n}(\bb r_n)\C P_{U(n)}(r_n|\bb r_{n-1})\C V_{SU(n)}(\bb r_{n-1}) \dots \C U_{2,1}(r_1)\C P_{U(1)}(r_1)\C V_{SU(1)}\C U_{1,0} \big[\rho_{SB}(t_0)\otimes\rho_{U(n)}(\bb r_{n-1})\otimes\dots\otimes\rho_{U(1)}\big]\big\}.    \nonumber
 \end{align}
\end{widetext}
Except for the unitary system-bath evolution superoperator $\C U$ (where the subscripts denote time intervals), 
subscripts are used to denote the Hilbert space on which the respective (super-) operator is acting. In this respect, 
the joint space of all $n$ units is denoted by $U(\bb n)$. Notice that $\C V_{SU(n)}(\bb r_{n-1})$ depends on all 
previous outcomes $\bb r_{n-1}$, but due to causality it cannot depend on the $n$'th outcome $r_n$. The same holds true 
for the initial state $\rho_{U(n)}(\bb r_{n-1})$ of the $n$'th unit and also the chosen projection operator 
$\C P_{U(n)}(r_n|\bb r_{n-1})$ can depend on $\bb r_{n-1}$. Therefore, the external agent has all the freedom 
she needs to engineer a desired control operation $\C A(r_n|\bb r_{n-1})$. By construction, after tracing out the 
units, we obtain the process tensor for the system 
$\mf T[\C A(r_n|\bb r_{n-1}),\dots,\C A(r_1)] = \mbox{tr}_{U(\bb n)}\{\tilde\rho_{SU(\bb n)}(t,\bb r_n)\}$. 
As it is in most situations obvious from the context which superoperator acts on which object living in which space, 
we will usually drop the subscripts $S, U(n), \dots$ on superoperators. 

\begin{figure}
 \centering\includegraphics[width=0.38\textwidth,clip=true]{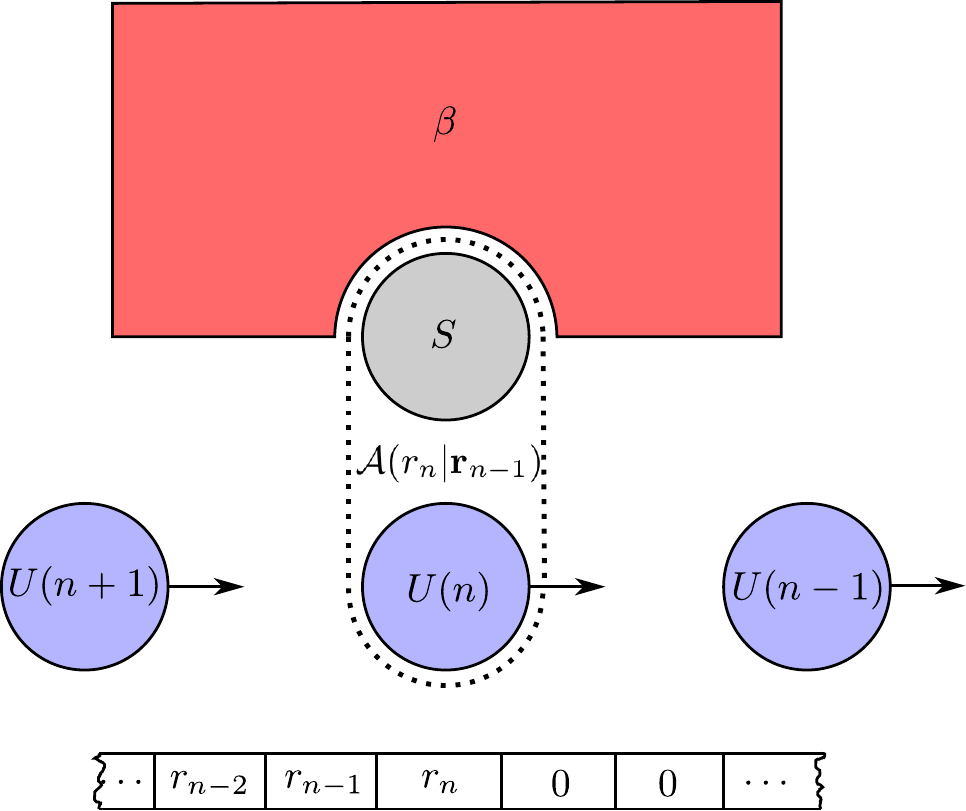}
 \label{fig setup} 
 \caption{Sketch of the setup: A system $S$ (grey circle) is in contact with a bath $B$ (red box, later taken to be at 
 inverse temperature $\beta$) undergoing in general dissipative dynamics. The evolution of the open quantum system 
 is interrupted at times $t_n$ by control operations $\C A(r_n|\bb r_{n-1})$, which are triggered by the interaction 
 with an external ancilla system called the unit $U(n)$ (blue circles). Each control operation has an outcome $r_n$, 
 which is recorded in a memory (e.g., a tape of bits) and future control operations are allowed to depend on previous 
 outcomes. The memory for future outcomes is set in a standard state `0'. }
\end{figure}

\section{Operational quantum stochastic thermodynamics}
\label{sec general stoch thermo}

\subsection{Preliminary considerations}
\label{sec preliminary considerations}

The process tensor is a formal object which does not make any assumptions about the system-bath dynamics. On the 
contrary, the standard ensemble averaged (or better: \emph{unmeasured}) framework of quantum thermodynamics relies on 
a weakly coupled, memoryless and macroscopic bath~\cite{SpohnLebowitzAdvChemPhys1979, AlickiJPA1979, LindbladBook1983, 
KosloffEntropy2013}. In this section we remain within this weak-coupling paradigm 
because possible extensions beyond the weak-coupling and Markovian assumption have only recently raised attention 
(see also Sec.~\ref{sec outlook}). Furthermore, we consider in this section only the case of a single heat bath at 
inverse temperature $\beta = 1/T$ ($k_B\equiv1$). The extension to multiple heat baths is subtle, see 
Sec.~\ref{sec final remarks}. 

Let us focus on the interval $(t_{n-1},t_n)$ (excluding the control operations at the boundaries) and let $\rho_S(t)$ be 
the system state at time $t\in(t_{n-1},t_n)$ (which is later on allowed to depend on $\bb r_{n-1}$). The state functions 
internal energy and system entropy for an arbitrary system state $\rho_S(t)$ are defined as
\begin{align}
 E_S(t) &\equiv  \mbox{tr}_S\{H_S(\lambda_t)\rho_S(t)\}, \\
 S_S(t) &\equiv  S_\text{vN}[\rho_S(t)].
\end{align}
According to the first law, the change in system energy $\Delta E_S^{(n)} \equiv E_S(t_n^-) - E_S(t^+_{n-1})$ can be 
split into heat and work, $\Delta E_S^{(n)} = W_S^{(n)} + Q_S^{(n)}$, by defining 
\begin{align}
 W_S^{(n)}  &\equiv \int_{t^+_{n-1}}^{t^-_n} dt \mbox{tr}_S\left\{\frac{\partial H_S(\lambda_t)}{\partial t}\rho_S(t)\right\},  \label{eq W standard}  \\
 Q_S^{(n)}  &\equiv \int_{t^+_{n-1}}^{t^-_n} dt \mbox{tr}_S\left\{H_S(\lambda_t)\frac{\partial \rho_S(t)}{\partial t}\right\}.  \label{eq Q standard}
\end{align}
Furthermore, the validity of the second law can be also derived and states that the entropy 
production is always positive: 
\begin{equation}\label{eq 2nd law standard}
 \Sigma^{(n)} \equiv \Delta S_S^{(n)} - \beta Q_S^{(n)} \ge 0,
\end{equation}
where $\Delta S_S^{(n)} \equiv S_S(t^-_n) - S_S(t^+_{n-1})$. 

Our goal in the rest of this section is to find definitions of internal energy, work, heat and system entropy 
along a single trajectory, where a trajectory is \emph{defined} by the observed sequence of outcomes $\bb r_n$. 
The sought-after definitions are required to be intuitively meaningful, to fulfill the first law at the trajectory level 
and the second law on average. Further appeal to our definitions will be added in 
Secs.~\ref{sec example},~\ref{sec limiting cases} and~\ref{sec final}. 

Note that, after tomographic reconstruction of the process tensor (see Sec.~\ref{sec process tensor}), we know the 
conditional system states $\rho_S(t_n^\pm,\bb r_n)$ only right before or right after the $n$'th control operation, but 
not in between for $t_{n-1} < t < t_n$. To compute the work~(\ref{eq W standard}) or heat~(\ref{eq Q standard}) 
in between two control operations, additional \emph{theoretical} input is required, e.g., by solving the master 
equation for the system or by other forms of inference. This ensures that we recover the standard weak coupling 
framework of quantum thermodynamics in absence of any control operations (see Sec.~\ref{sec standard quantum thermo}). 
Nevertheless, as it increases the computational effort, we present in 
Sec.~\ref{sec quantum stoch thermo without theory} possible ways to avoid any additional theory input. 

For definiteness, we aim at a stochastic thermodynamic description in the time interval $(t_{n-1},t_n]$ starting shortly 
after the $(n-1)$'th control operation and ending shortly after the $n$'th control operation. The change in any state 
function $X$ over the complete interval is denoted by $\Delta X^{(n]}$, whereas $\Delta X^{(n)}$ denotes the change in 
$(t_{n-1},t_n)$ (excluding the $n$'th control operation) and $\Delta X^\text{ctrl}$ the change due to the control 
operation only. Changes in the respective time intervals of any quantity which is not a state function are denoted 
without a delta ($X^{(n]}$, $X^{(n)}$ or $X^\text{ctrl}$). 

\subsection{Stochastic energy and first law}
\label{sec stochastic energetics}

To formulate the first law at the trajectory level correctly, we need to take into account the internal energy of the 
system and all units. Thus, we define the trajectory dependent internal energy 
\begin{equation}\label{eq internal energy}
  E_{SU(\bb n)}(t,\bb r_n) \equiv \mbox{tr}_{SU(\bb n)}\left\{H_{SU(\bb n)}(\lambda_t,\bb r_n) \rho_{SU(\bb n)}(t,\bb r_n)\right\},
\end{equation}
where $H_{SU(\bb n)}(\lambda_t,\bb r_n) = H_S(\lambda_t,\bb r_n) + \sum_{i=1}^n H_{U(i)}$ denotes the sum of the 
system and all unit Hamiltonians. Since the Hamiltonian is additive, the internal energy splits into its marginal 
contributions in the obvious way, 
\begin{equation}
 E_{SU(\bb n)}(t,\bb r_n) = E_S(t,\bb r_n) + \sum_{i=1}^n E_{U(i)}(t,\bb r_n).
\end{equation}
Notice that it is always simple to get rid of the units in the energetic description by assuming that 
$H_{U(i)} \sim 1_{U(i)}$. However, already the energetic changes of the units can bear some interesting non-trivial 
features. For instance, it is not sufficient to consider only the actual $n$'th unit in the energetic balance: 
in our general theory the energy of previous units can change even though they are \emph{physically decoupled} from the 
system. This phenomenon does not necessarily require quantum entanglement and simply occurs because our state of 
knowlegde about past units $U(i<n)$ can change depending on the outcome $r_n$ (see below). 

In absence of any control operations, the first law simply follows from the preceeding subsection and reads 
\begin{equation}\label{eq 1st law standard}
 \Delta E_S^{(n)}(\bb r_{n-1}) = W_S^{(n)}(\bb r_{n-1}) + Q_S^{(n)}(\bb r_{n-1}),
\end{equation}
because the marginal state of the units does not change and hence, $\Delta E_{U(i)} = 0$ for all $i$. Note that 
the work $W_S^{(n)}(\bb r_{n-1})$ and heat $Q_S^{(n)}(\bb r_{n-1})$ depend on previous outcomes $\bb r_{n-1}$ for two 
reasons: first, the initial system state $\rho_S(t_{n-1}^+,\bb r_{n-1})$ depends on it, and second, the Hamiltonian 
$H(\lambda_t,\bb r_{n-1})$ can be a function of it in case we apply feedback control. 

The first law during the control operation at time $t_n$ is more interesting as the internal energy of both, system and 
units, can change. In total, the energetic cost $E^\text{ctrl}$ of the control operation is defined by 
\begin{equation}\label{eq 1st law ctrl}
E^\text{ctrl}(t_n,\bb r_n) \equiv \Delta E^\text{ctrl}_S(t_n,\bb r_n) + \sum_{i=1}^n \Delta E^\text{ctrl}_{U(i)}(t_n,\bb r_n).
\end{equation}
It is not a state function and can be split into a work and heat like contribution, 
\begin{equation}
 E^\text{ctrl}(t_n,\bb r_n) = W^\text{ctrl}(t_n,\bb r_{n-1}) + Q^\text{ctrl}(t_n,\bb r_n).
\end{equation}
This splitting stems from the convention we used to implement the control operation $\C A(r_n|\bb r_{n-1})$ in the 
repeated interaction framework: we first applied the unitary operation $\C V(\bb r_{n-1})$ to the joint system-unit 
state and afterwards measured the unit via $\C P(r_n)$. In general, we therefore use the 
definitions 
\begin{widetext}
 \begin{align}
  W^\text{ctrl}(t_n,\bb r_{n-1})    &=  \mbox{tr}_{SU(\bb n)}\left\{H_{SU(\bb n)}(\lambda_n,\bb r_{n-1})\left[\C V(\bb r_{n-1})\rho_{SU(\bb{n})}(t_n^-,\bb r_{n-1}) - \rho_{SU(\bb{n})}(t_n^-,\bb r_{n-1})\right]\right\}, \label{eq W ctrl}  \\
  Q^\text{ctrl}(t_n,\bb r_n)        &=  \mbox{tr}_{SU(\bb n)}\left\{H_{SU(\bb n)}(\lambda_n,\bb r_{n-1})\left[\rho_{SU(\bb{n})}(t_n^+,\bb r_n) - \C V(\bb r_{n-1})\rho_{SU(\bb{n})}(t_n^-,\bb r_{n-1})\right]\right\} \label{eq Q ctrl}
 \end{align}
\end{widetext}
with $\lambda_n \equiv \lambda_{t_n}$. Notice that the work-like contribution does not depend on the actual measurement 
outcome $r_n$ and corresponds to the energetic changes caused by a reversible (unitary) operation. The meaning of the 
heat injected during the control operation $Q^\text{ctrl}(t_n,\bb r_n)$ will be discussed further below, but we remark 
that a very similar construction was called `quantum heat' in Ref.~\cite{ElouardEtAlQInf2017}. 
A difference, which turns out to be crucial, is the fact that Elouard \emph{et al.}~applied this definition for 
the system only without including the unit in the description~\cite{ElouardEtAlQInf2017}, which causes different 
interpretations. Furthermore, we are more cautious and do not call it `quantum' heat. For further discussion on this 
topic see Sec.~\ref{sec projective measurements}. 

For now, let us notice that both quantities have some additional important properties. First of all, both can be split 
additively into changes affecting the system or the units, 
\begin{align}
 W^\text{ctrl}(\bb r_{n-1}) &=  W^\text{ctrl}_S(\bb r_{n-1}) + \sum_{i=1}^n W^\text{ctrl}_{U(i)}(\bb r_{n-1}),  \\
 Q^\text{ctrl}(\bb r_n) &=  Q^\text{ctrl}_S(\bb r_n) + \sum_{i=1}^n Q^\text{ctrl}_{U(i)}(\bb r_n).
\end{align}
Especially, the part affecting the system can be expressed solely in terms of the control operation 
$\C A(r_n|\bb r_{n-1})$ and its average $\C A_n \equiv \sum_{r_n} \C A(r_n|\bb r_{n-1})$ and is thus independent 
of the details of the unit $U(n)$, see also Ref.~\cite{StrasbergWinterArXiv2019}. Specifically, 
\begin{align}
 & W_S^\text{ctrl}(\bb r_{n-1}) \label{eq W ctrl S only} \\
 & ~= \mbox{tr}_S\{H_S(\lambda_n,\bb r_{n-1})(\C A_n - \C I)\rho_S(t_n^-,\bb r_{n-1})\}, \nonumber \\
 & Q_S^\text{ctrl}(\bb r_n) \label{eq Q ctrl S only} \\
 & ~= \mbox{tr}_S\left\{H_S(\lambda_n,\bb r_{n-1})\left[\frac{\C A(r_n|\bb r_{n-1})}{p(r_n|\bb r_{n-1})} - \C A_n\right]\rho_S(t_n^-,\bb r_{n-1})\right\}, \nonumber
\end{align}
where $p(r_n|\bb r_{n-1}) \equiv p(\bb r_n)/p(\bb r_{n-1})$. Furthermore, if we use that the marginal state of the 
previous $n-1$ units does not change during the unitary operation $\C V(\bb r_{n-1})$, we can deduce that the work 
actually depends only on the energetic changes of the system and the $n$'th unit, 
\begin{equation}\label{eq W ctrl simplified}
 W^\text{ctrl}(\bb r_{n-1}) = W^\text{ctrl}_S(\bb r_{n-1}) + W^\text{ctrl}_{U(n)}(\bb r_{n-1}).
\end{equation}
The previous properties allow us to deduce \emph{two separate} first laws for the control operation: 
\begin{align}
 \Delta E_S^\text{ctrl}(\bb r_n)            &=  W^\text{ctrl}_S(\bb r_{n-1}) + Q^\text{ctrl}_S(\bb r_n),    
 \label{eq 1st law ctrl local S}    \\
 \Delta E_{U(\bb n)}^\text{ctrl}(\bb r_n)   &=  W^\text{ctrl}_{U(n)}(\bb r_{n-1}) + Q^\text{ctrl}_{U(\bb n)}(\bb r_n).
\end{align}
Finally, we can deduce that the \emph{average} heat injected into the system or the previous units 
$U(i)$ ($i<n$) is always zero. Specifically, 
\begin{equation}\label{eq prop Q ctrl}
 Q^\text{ctrl}_{S,U(i<n)}(t_n,\bb r_{n-1}) \equiv \sum_{r_n} p(r_n|\bb r_{n-1})Q_{S,U(i<n)}^\text{ctrl}(\bb r_n) = 0,
\end{equation}
Note that this equation implies 
$Q^\text{ctrl}_{S,U(i<n)}(t_n) = \sum_{\bb r_n} p(\bb r_n)Q_S^\text{ctrl}(\bb r_n) = 0$. In contrast, for 
the actual unit we have $Q^\text{ctrl}_{U(n)}(t_n) = 0$ if and only if $[H_{U(n)},P(r_n|\bb r_{n-1})] = 0$. We remark 
that it also appears reasonable to call $Q^\text{ctrl}$ `heat' because the emergence of a projector $\C P(r_n)$ 
requires in a microscopic picture to couple the unit to some macroscopic and classical device, which allows the unit to 
lose information irreversibly due to dissipation and decoherence~\cite{ZurekRMP2003}. This last phenomenological step 
in quantum measurement theory is sometimes refered to as the `Heisenberg cut'~\cite{WisemanMilburnBook2010}. It 
necessarily entails a certain level of arbitrariness because we do not explicitly model the microscopic interaction 
between the unit and the final classical environment. It therefore remains unclear how far any notion of temperature is 
associated to the heat $Q^\text{ctrl}$ and we will investigate this in the next section further. 

To conclude, after adding the first laws with and without control operation together, we obtain for the changes over a 
complete interval  
\begin{equation}\label{eq 1st law}
 \Delta E_S^{(n]}(\bb r_n) + \Delta E_{U(\bb n)}^{(n]}(\bb r_n) = W^{(n]}(\bb r_{n-1}) + Q^{(n]}(\bb r_n),
\end{equation}
where we can split the work and heat into 
$W^{(n]}(\bb r_{n-1}) = W^\text{ctrl}(\bb r_{n-1}) + W_S^{(n)}(\bb r_{n-1})$ and 
$Q^{(n]}(\bb r_{n}) = Q^\text{ctrl}(\bb r_{n}) + Q_S^{(n)}(\bb r_{n-1})$. If we assume trivial Hamiltonians for the 
units ($H_{U(i)}\sim1_U$), we get the simplified first law 
\begin{equation}
 \Delta E_S^{(n]}(\bb r_n) = W_S^{(n]}(\bb r_{n-1}) + Q_S^{(n]}(\bb r_n).
\end{equation}
For the entropic balance, it will be in general not that simple. 

\subsection{Stochastic entropy and second law}
\label{sec stochastic entropy}

To account for all entropic changes, we do not only need to consider the system and all units, but also the entropy 
of the outcomes $\bb r_n$ stored in a classical memory (see Fig.~\ref{fig setup}). This is a crucial point, which 
distinguishes our theory from standard stochastic thermodynamics where the entropic contribution of the measurement 
results is neglected (this will play an important role in Sec.~\ref{sec standard stoch thermo}). In general, 
the process tensor depends explicitly on the knowledge of $\bb r_n$, which cannot be neglected. Furthermore, it is 
important to also keep the past information of all previous units $U(i<n)$ and outcomes $\bb r_{n-1}$ because we 
explicitly allow the current unit and Hamiltonian to depend on all earlier outcomes (this is, for instance, essential if 
we apply time-delayed feedback control). Thus, we define the stochastic entropy of the process as 
\begin{equation}\label{eq entropy}
 S_{SU(\bb n)}(t,\bb r_n) \equiv - \ln p(\bb r_n) + S_\text{vN}[\rho_{SU(\bb n)}(t,\bb r_n)].
\end{equation}
Note that the probability $p(\bb r_n)$ of a particular trajectory can be straightforwardly computed from knowing 
the unnormalized state of the system, see Eq.~(\ref{eq probability process tensor}). If this state is not known, 
evaluation of Eq.~(\ref{eq entropy}) requires knowledge of many experimentally sampled trajectories first. Notice 
that the same is true for the definition of the trajectory dependent entropy in classical stochastic 
thermodynamics~\cite{SeifertPRL2005, SeifertRPP2012, VandenBroeckEspositoPhysA2015}. 

Next, we define the entropy production along a single trajectory over a time interval $(t_{n-1},t_n]$ by adding to the 
change in stochastic entropy the heat flow into the \emph{system}, 
\begin{equation}\label{eq entropy production}
 \Sigma^{(n]}(\bb r_n) \equiv \Delta S^{(n]}_{SU(\bb n)}(\bb r_n) - \beta Q_S^{(n]}(\bb r_n).
\end{equation}
As in classical stochastic thermodynamics, this expression can have either sign, but on average it is always positive 
as we will show below. Crucially, we have 
only taken into account the heat accociated with system changes whereas we did not include $Q^\text{ctrl}_{U(\bb n)}$ 
in the entropic balance. This will give us the correct result in all limiting cases and, if we use the commonly made 
assumption that $H_{U(i)} \sim 1_{U(i)}$, we anyway have $Q^\text{ctrl}_{U(\bb n)} = 0$ always. Furthermore, as we do 
not microscopically model the final projective measurement step of the units, it is also unclear which temperature we 
should associate to heat changes in the units and hence, including $Q^\text{ctrl}_{U(\bb n)}$ in the second law would 
necessarily imply some ambiguity. While these are all good \emph{a posteriori} arguments, the question whether there 
exist good \emph{a priori} arguments remains. 

To show the positivity of the average entropy production, it is useful to split it into two contributions similar to the 
first law: 
\begin{equation}\label{eq ent prod splitting}
 \Sigma^{(n]}(\bb r_n) \equiv \Sigma^\text{ctrl}(\bb r_n) + \Sigma^{(n)}(\bb r_{n-1})
\end{equation}
with 
\begin{align}
 \Sigma^\text{ctrl}(\bb r_n) =&~ 
 \Delta S^\text{ctrl}_{SU(\bb n)}(\bb r_n) - \beta Q^\text{ctrl}_S(\bb r_n), \label{eq ent prod ctrl}    \\
 \Sigma^{(n)}(\bb r_{n-1}) =&~ 
 \Delta S^{(n)}_{SU(\bb n)}(\bb r_{n-1}) - \beta Q_S^{(n)}(\bb r_{n-1}).  \label{eq ent prod no control}
\end{align}
We will now show that the second contribution $\Sigma^{(n)}$ is positive even along a single trajectory, whereas the 
first contribution $\Sigma^\text{ctrl}$ is positive only on average. 

To show $\Sigma^{(n)}(\bb r_{n-1}) \ge 0$ we will use Eq.~(\ref{eq 2nd law standard}), which holds for an arbitrary 
initial state $\rho_S(t_{n-1}^+,\bb r_{n-1})$, together with the fact that the system evolution in between two control 
operations can be described by a CPTP map independent of the initial state. This is true within the weak couling 
paradigm of quantum thermodynamics~\cite{SpohnLebowitzAdvChemPhys1979, AlickiJPA1979, LindbladBook1983, 
KosloffEntropy2013} where the time evolution is governed by a (possible time dependent) master equation in 
Lindblad-Gorini-Kossakowski-Sudarshan form. Let us denote the CPTP map by $\C E_n = \C E_n(\bb r_{n-1})$ such that 
\begin{equation}\label{eq dynamical map}
 \rho_S(t_{n}^-,\bb r_{n-1}) = \C E_n \rho_S(t_{n-1}^+,\bb r_{n-1}). 
\end{equation}
The inequality $\Sigma^{(n)}(\bb r_{n-1}) \ge 0$ can then be derived along the following lines:  

\begin{widetext}
 First, by using the mutual information $I_{S:U(\bb n)}$ between the system and the stream of units, we can split the 
 change in joint entropy as 
 \begin{equation}
  \begin{split}
  \Delta S^{(n)}_{SU(\bb n)}(\bb r_n) =&~ 
    S_\text{vN}[\rho_S(t_n^-,\bb r_{n-1})] + S_\text{vN}[\rho_{U(\bb n)}(t_n^-,\bb r_{n-1})] - I_{S:U(\bb n)}(t_n^-)	\\
  & - S_\text{vN}[\rho_S(t_{n-1}^+,\bb r_{n-1})] - S_\text{vN}[\rho_{U(\bb n)}(t_{n-1}^+,\bb r_{n-1})] + I_{S:U(\bb n)}(t_{n-1}^+).
  \end{split}
 \end{equation}
 Since the marginal state of the units does not change under the action of the CPTP map $\C E_n$, their entropic 
 contribution cancels out and we can write in short 
 $\Delta S^{(n)}_{SU(\bb n)}(\bb r_{n-1}) = \Delta S_{S}^{(n)}(\bb r_{n-1}) - \Delta I^{(n)}_{S:U(\bb n)}(\bb r_{n-1})$.
 Let us now add the entropy flow $-\beta Q_S^{(n)}(\bb r_{n-1})$ into the bath to the entropy balance. From the 
 second law~(\ref{eq 2nd law standard}) we can then infer that 
 \begin{equation}
  \Delta S^{(n)}_{SU(\bb n)}(\bb r_{n-1}) -\beta Q_S^{(n)}(\bb r_{n-1}) \ge - \Delta I^{(n)}_{S:U(\bb n)}(\bb r_{n-1}).
 \end{equation}
 The positivity of the right hand side is then guaranteed by contractivity of relative entropy under CPTP 
 maps~\cite{UhlmannCMP1977, OhyaPetzBook1993}. More specifically, the following chain of (in)equalities applies: 
 \begin{equation}\label{eq mutual info no control}
  \begin{split}
   I_{S:U(\bb n)}(t_{n-1}^+,\bb r_{n-1}) 
   &= D[\rho_{SU(\bb n)}(t_{n-1}^+,\bb r_{n-1})\|(\rho_S\otimes\rho_{U(\bb n)})(t_{n-1}^+,\bb r_{n-1})]   \\
   &\ge D[\C E_n\rho_{SU(\bb n)}(t_{n-1}^+,\bb r_{n-1})\|\C E_n(\rho_S\otimes\rho_{U(\bb n)})(t_{n-1}^+,\bb r_{n-1})] 
   = I_{S:U(\bb n)}(t_n^-,\bb r_{n-1}),
  \end{split}
 \end{equation}
 where it was essential that $\C E_n$ acts only on $S$ and not on $U(\bb n)$. This concludes the proof of positivity 
 of $\Sigma^{(n)}(\bb r_{n-1})$. 
\end{widetext}

Next, we will show that $\Sigma^\text{ctrl}(\bb r_n)$ is positive on average. More specifically, we will show that 
\begin{equation}\label{eq positivity EP ctrl}
 \Sigma^\text{ctrl}(\bb r_{n-1}) \equiv \sum_{r_n} p(r_n|\bb r_{n-1}) \Sigma^\text{ctrl}(\bb r_n) \ge 0.
\end{equation}
If this holds, then it also follows that 
$\Sigma^\text{ctrl}(t_n) = \sum_{\bb r_n} p(\bb r_n) \Sigma^\text{ctrl}(\bb r_n) \ge 0$. 
After taking the average and using Eq.~(\ref{eq prop Q ctrl}), we are left with three terms 
\begin{align}
 \Sigma^\text{ctrl}(t_n,\bb r_{n-1})   =&~ S_\text{Sh}[p(r_n|\bb r_{n-1})]  \nonumber   \\
                                       &+  \sum_{r_n} p(r_n|\bb r_{n-1})S_\text{vN}[\rho_{SU(\bb n)}(t_n^+,\bb r_n)]    \nonumber \\
                                       &-  S_\text{vN}[\rho_{SU(\bb n)}(t_n^-,\bb r_{n-1})], \label{eq EP help}
\end{align}
where $S_\text{Sh}[p(r_n|\bb r_{n-1})]$ is the Shannon 
entropy of the conditional probability $p(r_n|\bb r_{n-1})$.\footnote{To be distinguished from the conventional 
conditional entropy given by $\sum_{\bb r_{n-1}} p(\bb r_{n-1}) S_\text{Sh}[p(r_n|\bb r_{n-1})]$. }
The positivity of $\Sigma^\text{ctrl}(t_n,\bb r_{n-1})$ then follows from combining two theorems in quantum 
measurement theory: 

\begin{lemma}
 Let $\rho$ be an arbitrary state, $\{P_n\}_n$ a set of positive operators fulfilling $\sum_n P^2_n = 1$, 
 $p_n = \mbox{tr}\{P_n \rho P_n\}$ the probability to obtain 
 outcome $n$ and $\rho^{(n)} = P_n\rho P_n/p_n$ the post-measurement state conditioned on outcome $n$. Then, 
 \begin{equation}\label{eq lemma}
  S_\text{\normalfont vN}(\rho) \le S_\text{\normalfont Sh}(p_n) + \sum_n p_n S_\text{\normalfont vN}(\rho^{(n)}). 
 \end{equation}
\end{lemma}

\begin{proof}
 We first use that for any such set $\{P_n\}_n$ (see Theorem~11 in Ref.~\cite{JacobsBook2014} or Ref.~\cite{Ando1989}) 
 \begin{equation}
  S_\text{vN}(\rho) \le S_\text{vN}\left(\sum_n p_n\rho^{(n)}\right),
 \end{equation}
 i.e., the average uncertainty after the measurement can only increase. Next, we use (see Theorem~11.10 in 
 Ref.~\cite{NielsenChuangBook2000} or Refs.~\cite{LanfordRobinson1968, NielsenPRA2001}) 
 \begin{equation}\label{eq lemma inequality 2}
  S_\text{vN}\left(\sum_n p_n\rho^{(n)}\right) \le S_\text{Sh}(p_n) + \sum_n p_n S_\text{vN}(\rho^{(n)}). 
 \end{equation}
 This concludes the proof. 
\end{proof}

We now apply the lemma to Eq.~(\ref{eq EP help}). If we identify  $\{P_n\}$ with $\{P(r_n|\bb r_{n-1})\}$ acting in 
the joint system-unit space, the probability $p_n$ with the conditional probability $p(r_n|\bb r_{n-1})$ and the 
post-measurement state $\rho^{(n)}$ with $\rho_{SU(\bb n)}(t_n^+,\bb r_n)$, we can deduce that 
\begin{equation}
 \begin{split}
  & S_\text{Sh}[p(r_n|\bb r_{n-1})] + \sum_{r_n} p(r_n|\bb r_{n-1})S_\text{vN}[\rho_{SU(\bb n)}(t_n^+,\bb r_n)] \\
  & \ge S_\text{vN}[\C V\rho_{SU(\bb n)}(t_n^-,\bb r_{n-1})].
 \end{split}
\end{equation}
Using that the von Neumann netropy is invariant under unitary transformations, we deduce our desired 
result. Finally, we remark that inequality~(\ref{eq lemma inequality 2}) was used before 
in quantum thermodynamics to show the positivity of the second law for a Maxwell demon employing quantum 
measurements~\cite{JacobsPRA2009}. 

\section{Real-time preparation and stabilization of photon number states via quantum feedback}
\label{sec example}

The ability to control individual quantum systems and to protect them against decoherence has become a key challenge 
in modern quantum science. Recently, experiments in quantum optics reported on the preparation and stabilization of 
photon number states by using quantum feedback control~\cite{SayrinEtAlNature2011, ZhouEtAlPRL2012}; see also 
Ref.~\cite{DotsenkoEtAlPRA2009} for preceeding theoretical work. We will here analyse Ref.~\cite{ZhouEtAlPRL2012} 
(which is very similar to Ref.~\cite{SayrinEtAlNature2011}) within the operational framework of quantum stochastic 
thermodynamics. We will give unique insights into the energetic and entropic balances of these experiments by using 
the time- and energyscales as reported in Ref.~\cite{ZhouEtAlPRL2012}. Moreover, we will see that the efficiency 
to \emph{prepare} a pure photon number state is surprisingly high in the experiment (the efficiency to 
\emph{stabilize} the pure photon state is zero). However, in order not to overburden 
the paper, we will leave some experimental imperfections aside. These additional imperfections are listed at the end 
of this section, but we emphasize already here that all of them can be included into the 
operational framework of quantum stochastic thermodynamics. We will further assume some familiarity of the reader 
with concepts from quantum optics, for a basic introduction see Ref.~\cite{HarocheRMP2013} and references therein. 
The notation is chosen close to the original references~\cite{SayrinEtAlNature2011, ZhouEtAlPRL2012}.

\subsection{Setup and dynamics}

\begin{figure}[t]
 \centering\includegraphics[width=0.45\textwidth,clip=true]{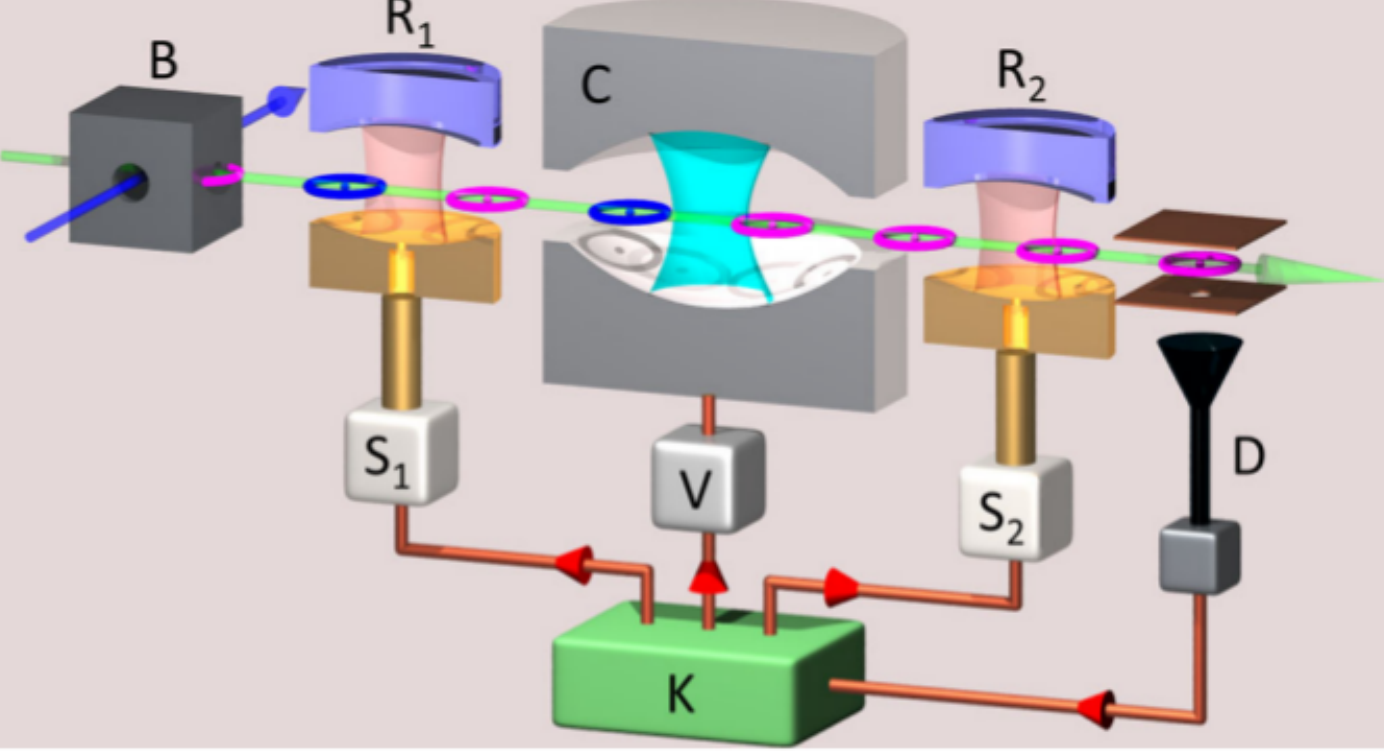}
 \label{fig exp setup} 
 \caption{Sketch of the experimental setup, compare also with Fig.~1 from Ref.~\cite{ZhouEtAlPRL2012}. We wish to control 
 the central microwave cavity {\tt C} by a beam of atoms prepared in {\tt B}. The atoms can be manipulated by the 
 Ramsey cavities {\tt R}$_1$ and {\tt R}$_2$ and read out by the detector {\tt D}. The measurement results are sent 
 to a controller {\tt K}, which decides in real time whether to send a sensor atom to measure the state of the cavity 
 (pink circles) or an emitter or absorber atom to manipulate the state of the cavity (blue circle). In the latter 
 case the atoms are brought into exact resonance with the cavity by applying a voltage {\tt V}.}
\end{figure}

A sketch of the experimental setup is shown in Fig.~\ref{fig exp setup}. The system we want to control is a 
superconducting Fabry-Perot cavity {\tt C} with Hamiltonian $\hbar\omega_c a^\dagger a$, where $a^\dagger$ and $a$ 
denote photon creation and annihilation operators and $\omega_c/2\pi = 51.1$~GHz is the experimentally measured 
frequency of the cavity (in this section we do \emph{not} set $\hbar\equiv1$). The cavity is coupled to an outside 
environment at temperature $T = 0.8$~K, which implies a Bose-Einstein distribution of 
$N_\text{th} = (e^{\beta\hbar\omega_c} - 1)^{-1} \approx 0.05$ (we also do not set $k_B\equiv 1$). The dynamics of the 
cavity are described by the master equation (in a rotating frame) 
\begin{equation} 
 \begin{split}
  \partial_t\rho_S(t)	&=	\C L_0\rho_S(t) \label{eq ME thermal}	\\
                        &\equiv \frac{1+N_\text{th}}{2T_c}\C D[a]\rho_S(t) + \frac{N_\text{th}}{2T_c}\C D[a^\dagger]\rho_S(t).
 \end{split}
\end{equation}
Here, the dissipator is defined as $D[a]\rho \equiv a\rho a^\dagger - \{a^\dagger a,\rho\}/2$ and the experimental 
cavity lifetime is $T_c = 65$~ms. 

Due to the interaction with the environment the cavity tends to thermalize to a Gibbs state, which, for the present 
parameters, means with probability 0.95 the vacuum state $|0\rangle$ with zero photons. The goal of the feedback loop 
is to reverse the effect of the dissipation and to stabilize a photon number state $|n\rangle = |n_t\rangle$ 
where $n_t > 0$ denotes the target number of photons in the following (we will choose $n_t = 2$ in the numerics). 
To achieve this goal, a beam of atoms created in {\tt B} via velocity selection and laser excitation is used. 
The atoms are repeatedly prepared at regular intervals of duration $T_a = 82~\mu$s and they leave {\tt B} with a 
velocity of $v = 250$~m/s. The interaction time of each atom with the cavity can be estimated as 
$t_\text{int} = \sqrt{\pi/2}\cdot\omega_0/v$ where $\omega_0 = 6$~mm is the  waist of the Gaussian cavity mode. This 
results in an interaction time of roughly $t_\text{int} \approx 30$~$\mu$s such that $T_c \approx 2000~t_\text{int}$. 
Thus, within very good approximation we can treat the interactions with the atoms as happening instantaneously as we 
have assumed in the formal development of our theory. Furthermore, the cavity lifetime is much larger than $T_a$ 
($T_c\approx 800~T_a$) such that we will approximate the dissipative time evolution in between two interactions by 
\begin{equation}\label{eq dissipative cavity}
 \C E = e^{\C L_0T_a} \approx 1 + \C L_0 T_a.
\end{equation}

\begin{figure}[t]
 \centering\includegraphics[width=0.29\textwidth,clip=true]{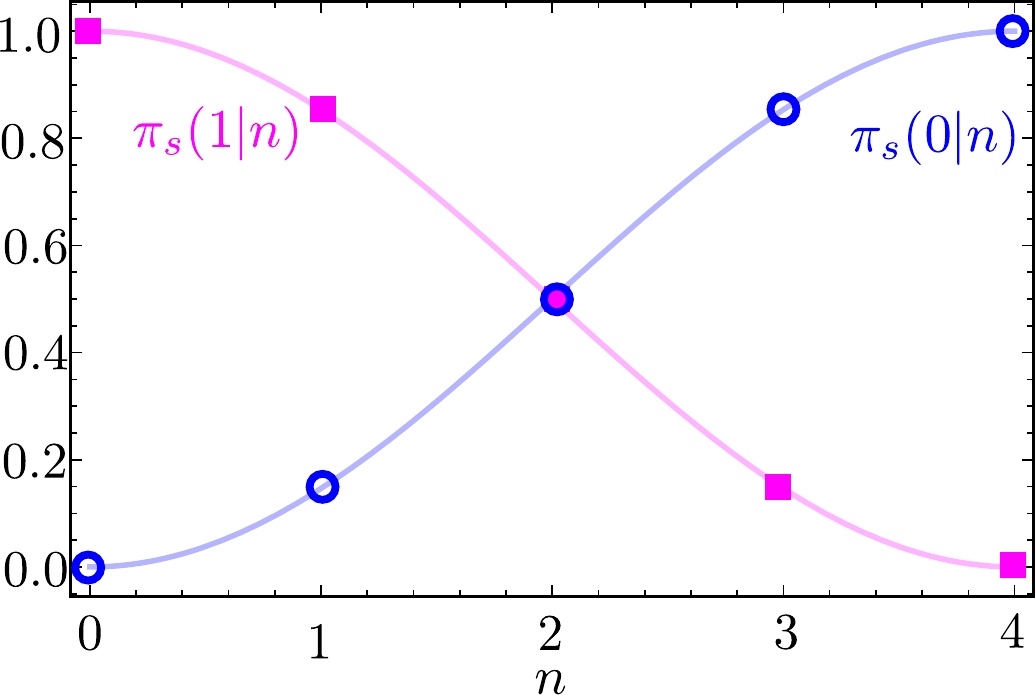}
 \label{fig cond prob} 
 \caption{Plot of the conditional probabilities as a discrete function of $n$: $\pi_s(0|n)$ (blue circles) and 
 $\pi_s(1|n)$ (pink filled squares). The solid lines serve only as a `guide for the eye'. }
\end{figure}

To counteract the dissipation by quantum feedback control, we first of all need to measure the state of the cavity. 
Importantly, this is done in a non-desctructive way without absorbing or emitting a photon using a modified Ramsey 
interferometry scheme. A brief theoretical description works as follows. First of all, the atoms are well-described as 
two-level systems with an energy gap $\hbar\omega_a\approx\hbar\omega_c$ close to the single photon energy in the 
cavity. We will denote the two levels as $|g\rangle$ and $|e\rangle$ for ground and excited state, respectively, albeit 
both states correspond to highly excited states of the atom, where the orbit of the outer electron is far away from 
the nucleus creating in turn a large dipole moment~\cite{HarocheRMP2013}. The atoms leave {\tt B} in the ground state 
$|g\rangle$ and are afterwards subjected to a $\pi/2$ pulse in cavity {\tt R}$_1$, which prepares them in the 
superposition $(|g\rangle+|e\rangle)/\sqrt{2}$. Due to an atom-cavity detuning of $\omega_a-\omega_c \approx 1.5~$MHz, 
the atom then interacts dispersively with the cavity field, which changes its state to 
$(|g\rangle+e^{i\phi(n)}|e\rangle)/\sqrt{2}$. Here, the $n$-dependent phase shift $\phi(n) = \Phi_0 n + \varphi_r$ 
is determined by the phase shift $\Phi_0$ per photon and the phase $\varphi_r$, which is adjustable in the Ramsey 
interferometer. Importantly, no energy is exchanged between the cavity and the atom during the interaction. 
Then, the atom is subjected to another $\pi/2$ pulse in cavity {\tt R}$_2$ and finally it is projectively measured in 
the detector {\tt D} revealing it either to be in the ground or excited state. The crux of the setup is that the  probability to find the atom in the ground or excited state depends on the number $n$ of photons in the cavity {\tt C}. 
If we denote by $r = 0$ the result corresponding to an atom found in the ground state and by $r = 1$ for an atom in an 
excited state, the conditional probability to obtain outcome $r$ given that there are $n$ photons in the cavity 
is\footnote{To deduce Eq.~(\ref{eq cond prob sensor}), we neglect experimental imperfections in the preparation 
and readout of the atoms and use in the notation of Ref.~\cite{ZhouEtAlPRL2012} 
$\pi_s(j|n) = [1+\cos(\Phi_0n+\varphi_r-j\pi)]/2$, where (opposite to our notation) $j=0$ ($j=1$) denotes an atom in 
the excited (ground) state. After taking this into account, setting the phase shift per atom to 
$\Phi_0 \approx \pi/4$~\cite{ZhouEtAlPRL2012} and adjusting the variable phase $\varphi_r$ of the Ramsey interferometer 
to the optimal value $\varphi_r+\Phi_0n_t = \pi/2$~\cite{ZhouEtAlPRL2012}, we obtain 
Eq.~(\ref{eq cond prob sensor}). } 
\begin{equation}\label{eq cond prob sensor}
 \pi_s(r|n) = \frac{1}{2}\left\{1 + \cos\left[\frac{\pi}{4}(n-n_t) + \frac{\pi}{2}(2r-1)\right]\right\}.
\end{equation}
For $n_t = 2$ this is exemplarily plotted in Fig.~\ref{fig cond prob} showing that it is clearly possible to 
distinguish between $n>n_t$, $n=n_t$ or $n<n_t$ photons in the cavity, but also demonstrating that we are far 
away from an ideal projective measurement of the cavity. 

\begin{figure*}[t]
 \centering\includegraphics[width=1.00\textwidth,clip=true]{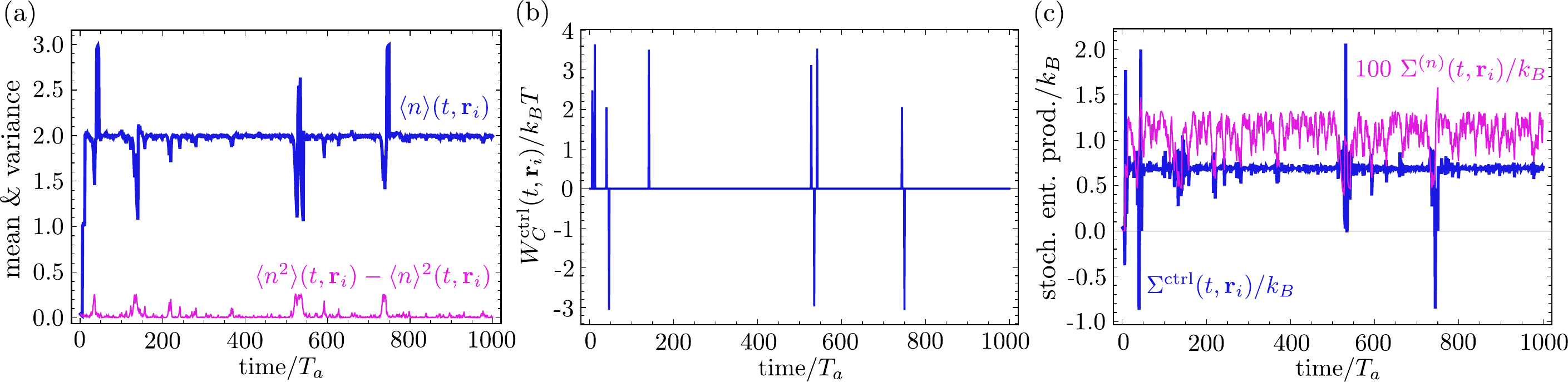}
 \label{fig plot stochastic} 
 \caption{Stochastic dynamics and thermodynamics of a single realization of the (numerical) experiment over 1000 
 time-steps. (a)~Conditional mean value $\lr{n}(t,\bb r_i)$, which fluctuates around the target number of $n_t = 2$ 
 photons (thick blue line on the top), and conditional variance $\langle n^2\rangle(t,\bb r_i) - \lr{n}^2(t,\bb r_i)$, 
 which is most of the time below $0.1$ (thin pink line on the bottom). (b)~(Dimensionless) work invested into the 
 control loop showing spikes exactly at the time when an emitter or absorber atom is sent into the cavity. 
 (c)~(Dimensionless) stochastic entropy production split according to Eq.~(\ref{eq ent prod splitting}) into the part 
 during the control operation, which can become temporarily negative (thick blue line), and the part in between the 
 control operations, which was upscaled by a factor of 100 for better visibility (thin pink line mostly on top).}
\end{figure*}

If we want to change the number of photons in the cavity, we can send an emitter or absorber atom into the cavity 
{\tt C}, which is either prepared in the excited or ground state respectively. For this purpose, the energy gap of the 
atoms is brought in exact resonance with the cavity by applying an external voltage {\tt V} (Stark shift) such 
that the atom-cavity dynamics is well-described by a Jaynes-Cummings Hamiltonian of the form (interaction picture) 
$h(a|e\rangle\langle g| + a^\dagger|g\rangle\langle e|)$. We will then ideally choose an effective interaction 
time $t_e = \pi/2h\sqrt{n_t}$ or $t_a = \pi/2h\sqrt{n_t+1}$ depending on whether we send an emitter or absorber atom 
respectively (this is slightly different from the experimental values). In the emitter case, the conditional 
probability to obtain outcome $r\in\{0,1\}$ and to observe a transition $n'\rightarrow n$ in the state of the cavity 
reads (compare, e.g., with Sec.~6.2.~in Ref.~\cite{ScullyZubairyBook1997}) 
\begin{equation}\label{eq cond prob emitter}
 \pi_e(r,n|n') = \sin^2\left(\frac{\pi}{2}\frac{\sqrt{n+r}}{\sqrt{n_t}}+\frac{\pi}{2}r\right)\delta_{n-1+r,n'},
\end{equation}
where $\delta_{n,n'}$ denotes the Kronecker delta. For the absorber case we get 
\begin{equation}\label{eq cond prob absorber}
 \pi_a(r,n|n') = \cos^2\left(\frac{\pi}{2}\frac{\sqrt{n+r}}{\sqrt{n_t+1}}+\frac{\pi}{2}r\right)\delta_{n+r,n'}.
\end{equation}
Note that, depending on the number $n$ of photons in the cavity, an absorber (emitter) atom will not always absorb 
(emit) a photon. 

Finally, it is important to realize that the atoms are detected \emph{time-delayed}, as indicated also in 
Fig.~\ref{fig exp setup}. This means that, before the $i$'th atom is registered with outcome $r_i$ at the detector 
{\tt D}, there have been already $d=5$ atoms which have been interacted or are about to interact with the cavity 
such that we cannot influence their initial state anymore. This point is important for the design of the feedback 
control law. In order to decide at time $t_i \equiv iT_a$ 
what kind of atom to send into the cavity, we can only use the state estimate 
$\rho_S(t_{i-d}^+,\bb r_{i-d})$ at time $t_{i-d} = (i-d)T_a$. Then, finally, the feedback control law is simply to 
sent an absorber atom as soon as we estimate $\sum_{n>n_t} p_n(t_{i-d},\bb r_{i-d}) > p_{n_t}(t_{i-d},\bb r_{i-d})$ 
and an emitter atom if we estimate $\sum_{n<n_t} p_n(t_{i-d},\bb r_{i-d}) > p_{n_t}(t_{i-d},\bb r_{i-d})$, where 
$p_n(t,\bb r_n)$ denotes the probability to have $n$ photons in the cavity at time $t$ given a measurement record 
$\bb r_n$. Otherwise we keep  measuring the system.  After each 
feedback operation we also wait $d$ time-steps before we apply the feedback control law again. This simple 
feedback control law is slightly different from the experiment, but as we will see now it works well. 

\subsection{Quantum stochastic thermodynamics}

In the previous section we have stated all necessary ingredients to apply our framework. The state of the cavity is 
conveniently described by the probability $p_n(t,\bb r_i)$ because coherences between different photon number states 
never play a role. The control operations $\C A(r_i|\bb r_{i-1})$ are either measurements or feedback operations (which 
can be emitative or absorbative).  Its effect on the cavity field, can be described by the conditional 
probabilities~(\ref{eq cond prob sensor}),~(\ref{eq cond prob emitter}) and~(\ref{eq cond prob absorber}). Due to the 
time-delay, we can set $\C A(r_i|\bb r_{i-1}) = \C A(r_i|\bb r_{i-d-1})$. Furthermore, the atoms always leave 
{\tt B} in the ground state $|g\rangle$ and they are always projected at the end of the interaction such that the 
sequence of outcomes $\bb r_i$ is simply a sequence of zeros and ones. Finally, the evolution in between two 
interactions is modeled by Eq.~(\ref{eq dissipative cavity}). Also numerically all parameters have been fixed in the 
previous section. 

\begin{figure*}[t]
 \centering\includegraphics[width=1.00\textwidth,clip=true]{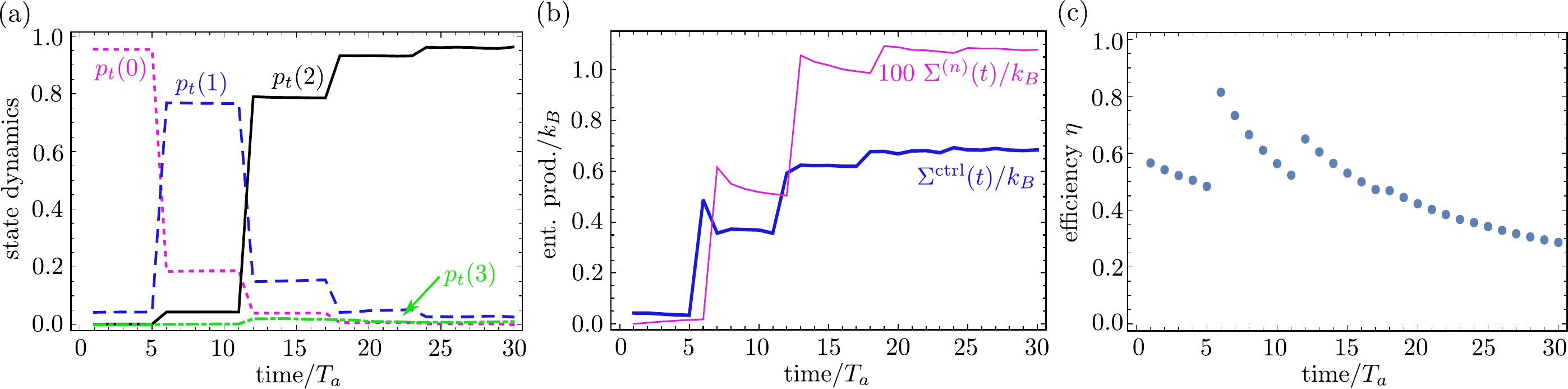}
 \label{fig plot average} 
 \caption{Averaged dynamics and thermodynamics of 2000 repetitions of the (numerical) experiment over 30 time-steps. 
 (a)~Dynamics of the cavity population displaying the probability to find zero photons (pink dotted line), 
 one photon (blue dashed line), two photons (black solid line) and three photons (green dash-dotted line). 
 (b)~(Dimensionless) entropy production again split according to Eq.~(\ref{eq ent prod splitting}). On average the 
 part associated to the control operations is now positive (thick blue line). The part in between the control operations 
 is again upscaled by factor 100 for better visibility (thin pink line). (c)~Time-dependent 
 efficiency~(\ref{eq efficiency Haroche}) of the experiment.}
\end{figure*}

Fig.~\ref{fig plot stochastic} shows various quantities for a single realization of the process over 1000 time 
intervals. The plot on the left shows the evolution of the conditional mean photon number 
$\langle n\rangle(t,\bb r_i) \equiv \sum_n n p_n(t,\bb r_i)$ (blue thick line) and the conditional variance 
$\langle n^2\rangle(t,\bb r_i) - \langle n\rangle^2(t,\bb r_i)$ (pink thin line). For perfect stabilization around 
$n_t = 2$ one would expect $\langle n\rangle(t,\bb r_i) = n_t$ and zero variance. As the plot shows, we are not far from 
that limit. The variance stays most of the time below 0.1 and only significantly deviates from it when our knowledge 
about the mean changes. This can be caused by an emission or absorbtion of a photon into or from the environment or by 
erroneous detection events as it is not perfectly possible to distinguish between 1, 2 or 3 photons (recall 
Fig.~\ref{fig cond prob}). Whenever our estimate about the mean changes significantly, the external agent performs a 
feedback control operation, which changes the energy of the cavity in a determinsitic way and entails a work 
cost $W_C^\text{ctrl}(\bb r_{i-1})$ as depicted in the middle of Fig.~\ref{fig plot stochastic}. Note that the work cost 
associated to the measurement of the cavity is zero, $W_C^\text{ctrl}(\bb r_{i-1}) = 0$, although there is always a 
work cost $W_A^\text{ctrl}(\bb r_{i-1})$ associated with the preparation of the atoms except for the case of 
absorbative feedback (see below; keep in mind that we use the subscripts $C$ and $A$ here instead of $S$ 
and $U$ as in the general Sec.~\ref{sec general stoch thermo}). As the plot demonstrates the experiment is not very 
costly in terms of the work invested into the cavity, which is roughly a few $k_BT$. Note that we also sometimes 
gain work as indicated by negative values and that also the work costs fluctuate due to the fact that the state of the 
system can be different at different times. What is more costly is the generation of the information needed to estimate 
the state of the cavity. This is shown in the plot on the right where the entropy production 
$\Sigma^\text{ctrl}(t,\bb r_i)$ (thick blue line) is roughly 0.7~$k_B$ at \emph{each} time step with some rare 
exceptions and strong fluctuations at the times where we perform feedback control operations. As a 
closer inspection reveals (not shown here), the main cause of this is the generation of information in the memory 
quantified by $-k_B\ln p(r_i|\bb r_{i-1})$. In comparison, the entropy produced in between two control operation 
$\Sigma^{(n)}(t,\bb r_i)$ (thin pink line), upscaled by a factor of 100, is much smaller than 
$\Sigma^\text{ctrl}(t,\bb r_i)$ due to the fact that in the short time interval $T_a$ not much is happening. Also note 
that $\Sigma^{(n)}(t,\bb r_i)$ is always positive as predicted by our theory. 

The previous observations are also confirmed by the average description. Fig.~\ref{fig plot average} shows the 
(thermo)dynamics for 30 time-steps averaged over 2000 numerical realisations. The first plot on the left depicts the 
time-evolution of the probabilities $p_0(t) = \sum_{\bb r_i} p_0(t,\bb r_i) p(\bb r_i)$ (dotted pink line), $p_1(t)$ 
(dashed blue line), $p_2(t)$ (solid black line) and $p_3(t)$ (dash-dotted green line) to have 0, 1, 2 or 3 
photons in the cavity. The effect of the time-delay $d=5$ can be clearly recognized as well as the success of the 
feedback loop to reach a pure photon state with probability $p_{n_t = 2}(t) \approx 0.96$. Note that the shown 
time interval of 30~$T_a\approx 0.04~T_c$ is too small to have a significant probability for a quantum jump 
induced by the environment, but to better see the impact of the time-delay we have decided to show here only a 
short time-window. Furthermore, the plot in the middle shows the entropy production $\Sigma^\text{ctrl}$ 
(thick blue line) and $\Sigma^{(n)}$ (thin pink line, scaled by a factor 100). In accordance with the previous plot we 
can conclude that the maintenance of the measurement and feedback loop is the thermodynamically most costly part with 
the most dominant contribution steming from the recording of the outcomes $\bb r_i$ in a classical memory (not shown). 
In addition, the plot also demonstrates that $\Sigma^\text{ctrl}$ is positive on average as predicted by our theory. 

Finally, the right plot in Fig.~\ref{fig plot average} shows the efficiency of the experiment in terms of 
generating a nonequilibrium state of the cavity with respect to the resources invested in the feedback loop. If we 
sum up the entropy production~(\ref{eq ent prod splitting}) in each time step and use the first 
laws~(\ref{eq 1st law standard}) and~(\ref{eq 1st law ctrl local S}), we can confirm that the average integrated 
entropy production after $N$ time-steps becomes 
\begin{equation}\label{eq 2nd law Haroche}
 \sum_{i=1}^N \Sigma^{(i]} = \frac{W_C^\text{tot}-\Delta F_C}{T} + k_B S_\text{Sh}[p(\bb r_N)] \ge 0.
\end{equation}
Here, $W_C^\text{tot}$ is the average integrated work invested into the cavity during the feedback loop (also see 
the middle plot in Fig.~\ref{fig plot stochastic}) and $\Delta F_C = F_C(NT_a) + k_B T\ln Z_C$ is the change in 
nonequilibrium free energy starting from a cavity state in equilibrium with partition function 
$Z_C = \mbox{tr}_C\{e^{-\beta\hbar\omega_c a^\dagger a}\}$. The average free energy after $N$ time-steps is computed 
by averaging over the energy and entropy of the cavity state, i.e., 
\begin{equation}
 F_C(NT_a) = \sum_{\bb r_N} p(\bb r_N) \{E_C(\bb r_N) - k_B T S_\text{vN}[\rho_C(\bb r_N)]\}.
\end{equation}
Finally, the last term in Eq.~(\ref{eq 2nd law Haroche}) denotes the entire average information content 
$S_\text{Sh}[p(\bb r_N)]$ associated to the outcomes of the experiment. It follows from the second 
law~(\ref{eq 2nd law Haroche}) that the following efficiency is bounded by one: 
\begin{equation}\label{eq efficiency Haroche}
 \eta \equiv \frac{\Delta F_C}{W_C^\text{tot} + k_B T S_\text{Sh}[p(\bb r_N)]} \le 1.
\end{equation}
As the right plot in Fig.~\ref{fig plot average} demonstrates, we can achieve remarkable high efficiencies peaked 
around values of $0.8$ and $0.65$ before they decay in the long run to zero. This decay is due to the fact 
that stabilizing the photon number state does not change its free energy anymore while the measurement and feedback 
loop still consumes resources. Thus, the \emph{preparation} of the photon number state is very efficient, but 
the \emph{stabilization} of it has by definition an efficiency zero.
While we focused on the average efficiency here, we remark that our framework is also ideally suited to 
study efficiency fluctuations~\cite{VerleyEtAlNC2014}. 

One might wonder why only the work invested into the cavity enters the definition 
of the efficiency~(\ref{eq efficiency Haroche}), but not the work $W_A^\text{ctrl}$ invested to prepare the state of 
the atoms. The latter is non-negligible: since the atoms leave {\tt B} in the ground state and the 
outgoing stream of atoms is roughly an equal mixture of atoms in the excited and ground state (which follows from 
Fig.~\ref{fig cond prob} once we have stabilized the state around $n_t$ photons), the work invested per atom is 
$W_A^\text{ctrl} \approx \hbar\omega_a/2$. This has its origin in the initial creation of the superposition in cavity 
{\tt R}$_1$.\footnote{This simple argument neglects the atoms used for the feedback control, which, however, 
constitute only a small fraction of the atoms used in the experiment. Furthermore, the proportion of outgoing atoms in 
the excited state is not exactly 0.5, but depends on the question whether the target photon number $n_t$ is above or 
below the thermal equilibrium value, which determines whether a state with $n_t$ photons tends to absorb or emit a 
photon into or from the environment.} However, what we are interested in here is how efficiently can we use a 
given amount of nonequilibrium resources (i.e., atoms in a pure state) to perform some task (creation of Fock states). 
That efficiency should be the same independent of, for instance, the question whether the atoms leaving {\tt B} are in 
the ground or excited state (in the latter case we would additionally extract work 
$W_A^\text{ctrl} \approx -\hbar\omega_a/2$ from 
the atoms). The second law of thermodynamics cares only about changes in entropy, which are zero for the incoming and outgoing stream of atoms.
In fact, the outgoing stream of atoms can be re-used again, e.g., for the next experiment. To make sense of this 
argument, it is important to note that the state of the atoms is \emph{not} in a mixture of ground and excited states 
because in each experimental run we exactly know the state of the atoms by looking at the measurement record $\bb r_N$. 
There is thus zero uncertainty associated to their state. 

The thermodynamic description would not be complete if we were to forget to mention that the 
experimental implementation also involves other costs, e.g., the cooling of the environment down to less than 1~K, the 
laser preparation of the atoms in {\tt B} or the electronics associated to the controller {\tt K}. What we have 
provided here is a minimal thermodynamic description of the system, which involves all essential contributions. 
Similar to other idealizations in thermodynamics, it is possible to imagine that the hidden thermodynamic costs of 
running the laboratory equipment can be made arbitrarily small in an ideal world. 

\subsection{Further experimental imperfections}

Finally, we  mention that the experiment is a little more complicated than described here. For instance, the number 
of atoms interacting with the cavity at a given time is not fixed to one, but rather Poisson distributed (with an 
average number of 0.6 atoms) such that there could be 0, 1 or 2 atoms per interaction. Furthermore, it can also 
happen that the detector {\tt D} misses to detect an atom. Those and other small imperfections are the reason why the 
experimentally observed probability $p_t(n_t)$ is around 0.8~\cite{ZhouEtAlPRL2012}. 

\section{Special cases}
\label{sec limiting cases}

\subsection{Projective measurement}
\label{sec projective measurements}

We start this section by considering the case of a single projective measurement. This is not only an illustrative 
example, but we will also need it in Secs.~\ref{sec two point meas approach} and~\ref{sec standard stoch thermo}. 

We denote the outcome of the projective measurement by $r$ and the associated projector by $|r\rl r|_S$. 
We assume no degeneracies in the measured observable here. Within our repeated interaction framework, we use a unit 
with Hilbert space of dimension $\dim\C H_U = \dim\C H_S$, the initial state is taken to be the pure state 
$|1\rl 1|_U$ and we assume a trivial unit Hamiltonian $H_U\sim 1_U$. The dynamical aspects of the Stinespring 
dilation are then fixed by the unitary $V$, 
\begin{equation}\label{eq unitary proj meas}
 V = \sum_{r,u} |r,u+r-1\rl r,u|_{SU}
\end{equation}
(where the sum in the `ket' has to be interpreted modulo $\dim\C H_U$), which is followed by a projective measurement 
of the unit in the basis $\{|r\rangle_U\}$. It is interesting to look at the states at the different steps of the 
process given an arbitrary initial system state $\rho_S(t^-) = \sum_s\lambda_s|s\rl s|_S$. After the unitary 
$V$ the marginal states of $\rho^*_{SU} \equiv V\rho_S(t^-)\rho_U V^\dagger$ are 
\begin{equation}
 \rho^*_S = \sum_r p(r)|r\rl r|_S, ~~~\rho^*_U = \sum_r p(r)|r\rl r|_U.
\end{equation}
Here, we have introduced the probability $p(r) = \sum_s |\langle r|s\rangle|^2 \lambda_s$ to obtain result $r$. 
After the projective measurement, the unnormalized state reads 
\begin{equation}
 \tilde\rho_{SU}(r,t^+) = p(r) |r,r\rl r,r|_{SU},
\end{equation}
from which it is easy to read of the marginal states and the average state after the control operation. 

Let us now look at the thermodynamic interpretation of a projective measurement within our framework. 
The work~(\ref{eq W ctrl}) and heat~(\ref{eq Q ctrl}) of the control operation become 
\begin{align}
 W_S^\text{ctrl}    &=  \sum_r p(r) \langle r|H_S|r\rangle - \sum_s \lambda_s \langle s|H_S|s\rangle,   \\
 Q_S^\text{ctrl}(r) &=  \langle r|H_S|r\rangle - \sum_{r'} p(r') \langle r'|H_S|r'\rangle.
\end{align}
We easily confirm $\sum_r p(r)Q_S^\text{ctrl}(r) = 0$. Moreover, the work vanishes whenever the measured basis 
coincides with the eigenbasis of the initial system state $\rho_S(t^-)$, albeit the fluctuating heat does not. Finally, 
we can also confirm Eq.~(\ref{eq lemma}), which boils down in this case to 
$S_\text{Sh}[p(r)] \ge S_\text{Sh}(\lambda_s)$. 

It is instructive to compare these results with the framework of Ref.~\cite{ElouardEtAlQInf2017}. In there, the 
quantum stochastic thermodynamics of projective measurements was also considered and the authors called the sum 
$W_S^\text{ctrl} + Q_S^\text{ctrl}(r)$ `quantum heat' and justified it by the fact that a quantum measurement 
is intrinsically stochastic unless the measured basis coincides with the basis of $\rho_S(t^-)$ (we add that only 
pure states were considered in Ref.~\cite{ElouardEtAlQInf2017}, thus leaving any classical uncertainty aside). 
Remarkably, we reach exactly the opposite conclusion on average: since $\sum_r p(r)Q_S^\text{ctrl}(r) = 0$, we infer 
that the average energetic change is purely work $W_S^\text{ctrl}$ instead of heat. 

This discrepancy can be traced back to the fact that we model the projective measurement in a larger space using 
Stinespring's theorem, which was not done in Ref.~\cite{ElouardEtAlQInf2017}. Remarkably, in this larger space we also 
called the energetic changes caused by the final measurement $P_U(r)$ `heat' (albeit not `quantum' heat because, as soon 
as classical uncertainty is considered too, it also plays a role, e.g., in classical stochastic thermodynamics, see 
Sec.~\ref{sec standard stoch thermo}). Thus, we applied a somewhat similar philosophy as Elouard 
\emph{et al.}~\cite{ElouardEtAlQInf2017}, but reached the opposite conclusion. This shows that the thermodynamic 
interpretation of a quantum measurement depends on where we put the Heisenberg cut. To defend the present approach, 
we want to highlight a number of key differences. 

First, by using Stinespring's theorem we pay duty to the fact that a quantum measurement does not happen spontaneously, 
but requires an \emph{active} intervention by the experimentalist, who brings two systems (the system to be measured and 
the detector) into contact. But bringing two different physical systems into contact, requires in general work 
(compare also with the `switching work' in Ref.~\cite{StrasbergEtAlPRX2017}). 

Second, our second law differs from the one derived in Ref.~\cite{ElouardEtAlQInf2017} as soon as 
multiple projective measurements are considered. In our case, the entropy production is on average given by the Shannon 
entropy of the entire sequence of measurement results 
$S_\text{Sh}[p(\bb r_n)] =  \sum_\ell S_\text{Sh}[p(r_\ell|\bb r_{\ell-1})]$ [with $p(r_1|r_0)\equiv p(r_1)$]. 
In Ref.~\cite{ElouardEtAlQInf2017} the entropy production is instead quantified by the Shannon entropy of the last 
measurement result only, $S_\text{Sh}[p(r_n)]$, and also the quantum heat does not enter their second law. 

Finally, we mention that the thermodynamic cost of quantum measurements was also explicitly studied 
elsewhere~\cite{SagawaUedaPRL2009, JacobsPRE2012, KammerlanderAndersSciRep2016, DeffnerPazZurekPRE2016, 
AbdelkhalekNakataReebArXiv2016}. In particular, Refs.~\cite{SagawaUedaPRL2009, JacobsPRE2012, 
KammerlanderAndersSciRep2016, AbdelkhalekNakataReebArXiv2016} reached similar conclusions by noting that performing 
a quantum measurement allows the external agent to extract work. Hence, the average energetic cost of the measurement should be counted as work. Also in a recent proposal of a Maxwell demon based only on projective measurements 
it was noted that the fields, which are controlled to implement the measurement, provide the energy for the 
demon~\cite{ElouardEtAlPRL2017}.

\subsection{The two-point measurement approach}
\label{sec two point meas approach}

The two-point measurement approach, which is closely related to the theory of full counting statistics, has become the 
primarily used approach to derive quantum fluctuation relations in various open quantum 
systems~\cite{EspositoHarbolaMukamelRMP2009, CampisiHaenggiTalknerRMP2011, SchallerBook2014}. While theoretically 
powerful, we already discussed the practical weakness of this approach in the introduction: experimental confirmations 
have been so far only achieved for work fluctuation relations in isolated systems~\cite{BatalhaoEtAlPRL2014, 
AnEtALNatPhys2015, CerisolaEtAlNatComm2017} or in electronic nanocircuits when the electrons behave 
according to a classical rate master equation~\cite{UtsumiEtAlPRB2010, KungEtAlPRX2012, SairaEtAlPRL2012}. 

We here critically re-examine the two-point measurement approach from a foundational perspective. We also view it in 
context of Ref.~\cite{PerarnauLlobetEtAlPRL2017}, which proves that there exists no measurement strategy of 
work, whose statistics fulfill (i) a quantum work fluctuation theorem and (ii) reproduce -- when averaged -- the 
unmeasured first law for arbitrary initial states. This important ``no-go theorem'' proves that quantum stochastic 
thermodynamics is distinctively different from its classical counterpart: it is in general impossible to make the 
averaged picture coincide with the unmeasured picture in quantum thermodynamics. Nevertheless, within our 
framework we will find that the no-go theorem does not apply in the sense that the `work' defined in the two-point 
measurement approach is not even work according to our framework. 

We consider the following standard scenario, where the unitary evolution of an isolated system is interrupted by two 
projective measurements. We assume that the projective measurements are described as in 
Sec.~\ref{sec projective measurements} with energetically neutral units. Since the system is isolated, we will also 
drop the subscript `S' on all quantities. 

First, the system is prepared in a Gibbs state such that 
\begin{equation}
 \rho(t_0^-) = \frac{e^{-\beta H(\lambda_0)}}{Z(\lambda_0)} 
 = \frac{1}{Z(\lambda_0)}\sum_{\epsilon_0} e^{-\beta\epsilon_0} |\epsilon_0\rl\epsilon_0|,
\end{equation}
where $Z(\lambda_0) = \mbox{tr}\{e^{-\beta H(\lambda_0)}\}$ denotes the partition function. Its internal energy is 
denotes by $E(t_0^-) = \mbox{tr}\{H(\lambda_0)\rho(t_0^-)\}$. 
Then, at time $t_0$ we projectively measure the energy and obtain outcome $r_0$, which is uniquely associated to 
one energy eigenvalue $\epsilon_0(r_0)$. Since the measurement basis coincides with the eigenbasis, the work during 
this measurement is zero. However, the internal energy clearly changes along a single trajecory and this is due to heat: 
\begin{equation}
 \epsilon_0(r_0) - E(t_0^-) = Q^\text{crtl}(r_0).
\end{equation}
In the next step we let the isolated system evolve according to an arbitrary time-dependent Hamiltonian $H(\lambda_t)$. 
The state at time $t_1>t_0$ is given by $|\psi(t,r_0)\rangle = U(t)|\epsilon_0(r_0)\rangle$ where $U(t)$ denotes the 
unitary time evolution operator generated by $H(\lambda_t)$. As the system is completely isolated, the change in 
internal energy is purely given by work: 
\begin{equation}\label{eq work TPM}
 \langle\psi(t_1,r_0)|H(\lambda_1)|\psi(t_1,r_0)\rangle - \epsilon_0(r_0) = W^{(1)}(r_0).
\end{equation}
Finally, there is another projective measurement in the eigenbasis of $H(\lambda_1)$ with outcome $r_1$, uniquely 
associated to some eigenenergy $\epsilon_1(r_1)$. The change in internal energy now has in general a work and a heat 
contribution: 
\begin{equation}
 \begin{split}
  & \epsilon_1(r_1) - \langle\psi(t_1,r_0)|H(\lambda_1)|\psi(t_1,r_0)\rangle \\
  & = W^\text{ctrl}(r_0) + Q^\text{ctrl}(r_1,r_0).
 \end{split}
\end{equation}
To derive an explicit form for it, we expand the prior state with respect to the final measurement basis: 
$|\psi(t_1,r_0)\rangle = \sum_{\epsilon_1} c_{\epsilon_1}|\epsilon_1\rangle$. Then, we obtain 
\begin{align}
 W^\text{ctrl}(r_0) &= \sum_{\epsilon_1} |c_{\epsilon_1}|^2 \epsilon_1 - \langle\psi(t_1,r_0)|H(\lambda_1)|\psi(t_1,r_0)\rangle \\
 Q^\text{ctrl}(r_1,r_0) &= \epsilon_1(r_1) - \sum_{\epsilon_1} |c_{\epsilon_1}|^2 \epsilon_1.
\end{align}
Both contributions differ from zero unless in the classical case where 
$|\psi(t_1,r_0)\rangle = |\epsilon_1(r_1)\rangle$. Thus, in that scenario it would be justified to call 
$Q^\text{ctrl}(r_1,r_0)$ ``quantum'' heat. 

Now, consider the probability for the sequence of outcomes 
\begin{equation}
 p(r_1,r_0) = |\langle\epsilon_1(r_1)|U(t)|\epsilon_0(r_0)\rangle|^2 \frac{e^{-\beta\epsilon_0(r_0)}}{Z(\lambda_0)}.
\end{equation}
It is a straightforward exercise to show that this probability distribution implies the so-called quantum work theorem 
or quantum Jarzynski equality, first derived in Refs.~\cite{PiechocinskaPRA2000, KurchanArXiv2000, TasakiArXiv2000}: 
\begin{equation}
 \begin{split}
  \lr{e^{-\beta[\epsilon_1(r_1)-\epsilon_0(r_0)]}} 
  &\equiv \sum_{\epsilon_1,\epsilon_0}p(r_1,r_0) e^{-\beta[\epsilon_1(r_1) - \epsilon_0(r_0)]} \\
                    &=      \frac{Z(\lambda_1)}{Z(\lambda_0)}.
 \end{split}
\end{equation}
Now, in analogue to the classical Jarzynski equality, the fluctuating quantity $\epsilon_1(r_1)-\epsilon_0(r_0)$ 
in the exponent was called `work' in the two-point measurement approach~\cite{EspositoHarbolaMukamelRMP2009, 
CampisiHaenggiTalknerRMP2011}. However, our framework reveals that 
\begin{equation}
 \epsilon_1(r_1)-\epsilon_0(r_0) = W^{(1)}(r_0) + W^\text{ctrl}(r_0) + Q^\text{ctrl}(r_1,r_0).
\end{equation}
That is to say, the fluctuating quantity in the exponent is not work alone. Hence, one better calls the 
quantum work theorem a quantum \emph{internal energy} theorem. 

We end this section by pointing out that we are not the first to criticize the notion of work within the two-point 
measurement approach. For instance, Deffner, Paz and Zurek also criticize this approach for not being 
``thermodynamically consistent as it does not account for the thermodynamic cost of 
measurements''~\cite{DeffnerPazZurekPRE2016}. Remarkably, they were able to derive a modified quantum Jarzynski equality 
for the work~(\ref{eq work TPM}) done in between the two projective measurements~\cite{DeffnerPazZurekPRE2016}. 

\subsection{The standard framework of quantum thermodynamics}
\label{sec standard quantum thermo}

If we perform no control operations at all, our framework obviously reproduces the standard framework of quantum 
thermodynamics mentioned at the beginning in Sec.~\ref{sec preliminary considerations}. This fact might seem so obvious 
that it is not worse to stress. However, it is important to realize that the standard framework of quantum 
thermodynamics cannot be recovered by performing an ensemble average over $p(\bb r_n)$, but only by deciding not to 
apply any control operation at all (apart from maybe preparing a certain initial state and reading out the final state). 
That is to say, in order to recover standard quantum thermodynamics, it is important to have a framework which can cope 
with incomplete information and allows to do `nothing' on the system. All previous frameworks of quantum stochastic 
thermodynamics, which rely on a perfectly measured system in a pure state, \emph{fail} to reproduce the picture without 
control operations because any measurement disturbs the process in general. In fact, in almost all previous works the 
notion of a stochastic entropy along a single trajectory is not even defined. To the best of the author's knowlegde, the 
only exceptions are Refs.~\cite{HorowitzPRE2012, HorowitzParrondoNJP2013} where, however, the definition of stochastic 
entropy \emph{depends} on the initial state chosen and therefore, needs to be adapted in each experiment.
The reason why classical stochastic thermodynamics reproduces the average picture 
(see Sec.~\ref{sec standard stoch thermo}) is the fact that there is always one fixed basis and no coherences are 
possible. The current framework therefore fills an important conceptual gap between quantum and classical 
stochastic thermodynamics. 

\subsection{The conventional repeated interaction framework}
\label{sec repeated interactions ensemble level}

The framework of repeated interactions gives rise to a generalized thermodynamic theory by realizing that the stream of 
external units can act in the most general scenario as a resource of nonequilibrium free energy, which encompasses many 
previously considered theories~\cite{StrasbergEtAlPRX2017} (see also Ref.~\cite{BarraSciRep2015} for important earlier 
work). However, the repeated interaction framework considered previously differs from our framework by avoiding to do 
any measurement on the units. In order to recover this thermodynamic framework, it is important to realize (as in 
Sec.~\ref{sec standard quantum thermo}) that a simple ensemble average of the process tensor over the outcomes 
$\bb r_n$ will \emph{not} do the job. The only correct way to recover previous results from our framework is to not 
perform any measurement, i.e., in the language of 
Sec.~\ref{sec process tensor repeated interactions} to choose the `projector' $P(r_n|\bb r_{n-1}) = 1_U$ 
throughout. In this case, the process tensor can be written as $\mf T[\C A_n,\dots,\C A_1]$ where $\C A_i$ is a CPTP map 
acting at time $t_i$. The control operations and hence, also the process tensor, do not depend on any outcome $\bb r_n$  
anymore (alternatively, one could say that each control operation at time $t_i$ has only one possible outcome). 
Furthermore, every incoming unit is decorrelated from the previous units as in Ref.~\cite{StrasbergEtAlPRX2017}. 

Our thermodynamic framework of the process tensor is therefore much more general and flexible than the previous 
framework apart from one important difference. In Ref.~\cite{StrasbergEtAlPRX2017} the units were allowed to interact 
with the system for a \emph{finite} duration whereas we here only consider instantaneous interactions (or more 
precicely,  interaction times where the effect of the bath can be neglected to leading order). From a thermodynamic 
point of view, this is not necessary. However, to be able to clearly distinguish between control operations on the 
system and system-bath dynamics, this assumption is necessary (compare with the discussion in 
Sec.~\ref{sec process tensor}). 

We now show that our thermodynamic framework is not in contradiction to the one of 
Ref.~\cite{StrasbergEtAlPRX2017}, if we avoid any measurements of the units. Since no quantity depends on 
$\bb r_n$ anymore, the internal energy is simply 
\begin{equation}
 E_{SU(\bb n)}(t) = E_S(t) + \sum_{i=1}^n E_{U(i)}(t).
\end{equation}
But the internal energy of all previous units $U(i<n)$ never enters the first law and thus, can be neglected. 
In fact, in absense of any control operation this is evident from Eq.~(\ref{eq 1st law standard}). 
During the control operations, because there is no final measurement, $Q^\text{ctrl}(t_n) = 0$ and only 
$W^\text{ctrl}(t_n)$ can differ from zero. But the work only depends on the state of the $n$'th unit 
and not on previous units [cf.~Eq.~(\ref{eq W ctrl simplified})]. Hence, the first law during the control operation 
becomes $W^\text{ctrl}(t_n) = \Delta E_S(t_n) + \Delta E_{U(n)}(t_n)$ because the marginal state of all other units 
does not change. We therefore obtain the same first law over one interaction period $(t_{n-1},t_n]$: 
\begin{equation}
 \Delta E_S^{(n]} + \Delta E_{U(n)}^{(n]} = W^\text{ctrl}(t_n) + W^{(n)} + Q^{(n)}.
\end{equation}
Finally, note that $W^\text{ctrl}(t_n)$ would be identified in context of Ref.~\cite{StrasbergEtAlPRX2017} 
with the switching work $W_\text{switch}$ required to turn on and off the system-unit interaction. 

We now turn to the second law. Without any outcomes $\bb r_n$ we obtain from Eq.~(\ref{eq entropy}) the entropy 
$S_{SU(\bb n)}(t) = S_\text{vN}[\rho_{SU(\bb n)}(t)]$. Again, this differs from Ref.~\cite{StrasbergEtAlPRX2017} by 
explicitly taking into account the joint entropy of \emph{all} units and the system. To recover 
Ref.~\cite{StrasbergEtAlPRX2017}, we start again with the situation without control operation. From 
Eq.~(\ref{eq 2nd law standard}) we know that $\Delta S_S^{(n)} - \beta Q_S^{(n)} \ge 0$ and, 
since the marginal unit states do not change, we can extend this to 
\begin{equation}\label{eq ent prod repeated interaction}
 \Delta S_S^{(n)} + \Delta S_{U(n)}^{(n)} - \beta Q_S^{(n)} \ge 0.
\end{equation}
Next, our second law during the control operation becomes 
\begin{equation}\label{eq ent prod ctrl repeated interaction}
 \Sigma^\text{ctrl}(t_n) = S_\text{vN}[\rho_{SU(\bb n)}(t_n^+)] - S_\text{vN}[\rho_{SU(\bb n)}(t_n^-)] = 0,
\end{equation}
because the von-Neumann entropy is invariant under unitary transformation. If we use the two facts that the 
unitary $\C V$ acts only locally on the system and the $n$'th unit and that the initial state of the unit is 
decorrelated from the system, we immediately confirm that Eq.~(\ref{eq ent prod ctrl repeated interaction}) can be 
rewritten as 
\begin{equation}
 \Sigma^\text{ctrl}(t_n) 
 = \Delta S^\text{ctrl}_\text{vN}(\rho_S) + \Delta S^\text{ctrl}_\text{vN}(\rho_U) - I_{S:U(n)}(t_n^+) = 0.
\end{equation}
Taking the mutual information to the other side of the equation and combining it with 
Eq.~(\ref{eq ent prod repeated interaction}), we can confirm for an entire interaction interval that 
\begin{equation}
 \Delta S_S^{(n]} + \Delta S_{U(n)}^{(n]} - \beta Q^{(n)} \ge I_{S:U(n)}(t_n^+) \ge 0.
\end{equation}
This reproduces the generalized second law from Ref.~\cite{StrasbergEtAlPRX2017}. The reason why the final mutual 
information between the system and the previous units was discarded in Ref.~\cite{StrasbergEtAlPRX2017} becomes clear 
by recalling that every unit which has already interacted with the system does not have the chance to interact with the 
system again. All final mutual information will therefore be lost. This is in contrast to the general framework 
developed here where we allowed for all kinds of feedback control. Under these more general circumstances, the 
remaining mutual information after the interaction represents a valuable thermodynamic resource, which cannot be 
neglected. 

\subsection{Standard classical stochastic thermodynamics}
\label{sec standard stoch thermo}

A tacitly made assumption in classical stochastic thermodynamics is the ability to measure perfectly (i.e., without 
error and without disturbance) the state of the system~\cite{SeifertRPP2012, VandenBroeckEspositoPhysA2015}. 
These assumptions can be completely overcome by using the operational approach to stochastic thermodynamics, but 
attention has to be payed to the fact that the classical version of Stinespring's theorem does not follow from the 
quantum version stated in Sec.~\ref{sec process tensor repeated interactions}~\cite{StrasbergWinterArXiv2019}. 

Here, we restrict ourselves to study the standard case of stochastic thermodynamics assuming perfect continuous 
measurements and no feedback control. We focus only on a classical discrete system, which makes random jumps between 
a finite set of states $s \in \{1,\dots,d\}$. Its dynamics are described by a rate master equation 
\begin{equation}\label{eq master equation stoch thermo}
 \frac{d}{dt}p_s(t) = \sum_{s'} R_{s,s'}(\lambda_t) p_{s'}(t).
\end{equation}
Here, $p_s(t)$ is the probability to find the system in state $s$ at time $t$, whose energy we denote by 
$H(s,\lambda_t)$ (dropping the subscript $S$ on $H$). The rate matrix $R_{s,s'}(\lambda_t)$ can depend on an external 
control parameter $\lambda_t$. It is required to fulfill the local detailed balance condition 
\begin{equation}
 \frac{R_{s,s'}(\lambda_t)}{R_{s',s}(\lambda_t)} = e^{-\beta[H(s,\lambda_t) - H(s',\lambda_t)]},
\end{equation}
which allows to link energetic changes in the system to entropic changes in the bath. Due to the assumptions of standard 
stochastic thermodynamics one knows at each time $t$ the state $s$ of the system without any uncertainty (denoted $s_t$ 
in the following). The stochastic energy and entropy at time $t$ is then defined by 
\begin{equation}\label{eq E S stoch thermo}
 E_\text{ST}(s_t) \equiv H(s_t,\lambda_t), ~~~ S_\text{ST}(s_t) \equiv -\ln p_{s_t}(t),
\end{equation}
where we used a subscript `ST' to denote definitions used in standard stochastic thermodynamics. Note that the 
stochastic entropy $S_\text{ST}(s_t)$ is determined by evaluating the solution of the rate master equation along a 
particular stochastic trajectory~\cite{SeifertPRL2005}. Work and heat for a sufficiently small time-step $dt$ are 
defined as\footnote{In stochastic thermodynamics, one usually writes $\delta W$ or $\dbar W$ to denote the 
infinitesimal character of the quanity. Often, one also denotes quantities defined for single trajectories with a 
small letter, e.g., $w$. We here decided to stick closer to our notation from Sec.~\ref{sec general stoch thermo} 
keeping in mind that we are only interested in small time steps $dt$. } 
\begin{align}
 W_\text{ST}(s_t)   &\equiv H(s_{t-dt},\lambda_t) - H(s_{t-dt},\lambda_{t-dt}), \label{eq W classical stoch thermo} \\
 Q_\text{ST}(s_t)   &\equiv H(s_t,\lambda_t) - H(s_{t-dt},\lambda_t) \label{eq Q classical stoch thermo}
\end{align}
such that $E_\text{ST}(s_t) - E_\text{ST}(s_{t-dt}) = W_\text{ST}(s_t) + Q_\text{ST}(s_t)$. Furthermore, using rather 
complicated algebraic manipulations, one can compute the change of stochastic entropy along a particular 
trajectory~\cite{SeifertPRL2005, SeifertRPP2012, VandenBroeckEspositoPhysA2015} (we will see below that evaluating the 
quantities in discrete time steps simplifies the algebra significantly). In the resulting expression it is then possible 
to single out a term related to the entropy production, which -- on average -- yields the always positive expression 
\begin{equation}\label{eq ent prod stoch thermo}
 \Sigma_\text{ST}(t) \equiv \Delta S_\text{ST}(t) - \beta Q_\text{ST}(t) \ge 0,
\end{equation}
where $\Delta S_\text{ST}(t) \equiv S_\text{Sh}[p_s(t)] - S_\text{Sh}[p_s(t-dt)]$ turns out to be the (infinitesimal) 
change in Shannon entropy of the solution $p_s(t)$ of the rate master equation and 
$Q_\text{ST}(t) = \sum_s H(s,\lambda_t) [p_s(t)-p_s(t-dt)]$ is the average heat entering the system per time step $dt$.

Our goal is now to show the following: (1)~how a perfect, non-disturbing measurement arises in our context; (2)~that 
we obtain identical expressions for the stochastic heat, work and internal energy in this limit; (3)~that we obtain a 
different expression for stochastic entropy, which yields a different, but meaningful second law; (4)~how the entropy 
production of standard stochastic thermodynamic arises in our context when we change the definition of stochastic 
entropy. 

(1) To obtain a perfect measurement, we can basically use the same steps as in Sec.~\ref{sec projective measurements}. 
We start with a classical probability $p_U = \delta_{u,1}$ and view the unitary~(\ref{eq unitary proj meas}) as a 
permutation matrix. Then, the result is that the state of the system gets copied onto the state of the unit. Next, we 
consider the limit where we measure the system \emph{continuously}, i.e., in small time-steps $dt = t_n - t_{n-1}$ such 
that the probability for a jump in each interval is very small: $R_{s,s'}(\lambda_t)dt \ll 1$. Furthermore, we assume 
that all units are identical and uncorrelated initially. In this limit, the sequence of measurement outcomes $\bb r_n$ 
is \emph{identical} to the state of the units, which is \emph{identical} to the trajectory taken by the system. This is 
the essence of a perfect classical and continuous measurement. As a consequence, the state of the system at time 
$t\ge t_n^+$ only depends on the last measurement outcome $r_n$, but not  on any of the previous outcomes 
$\bb r_{n-1}$. Furthermore, the state of the system during the interval $(t_{n-1},t_n]$ changes from 
$\bb p(t^+_{n-1},r_{n-1}) = |r_{n-1}\rangle$ at the beginning to 
$\bb p(t_n^-,r_{n-1}) = |r_{n-1}\rangle + dt\sum_s R_{s,r_{n-1}}(\lambda_t) |s\rangle$ shortly before the control 
operation and to $\bb p(t_n^+,r_n) = |r_n\rangle$ at the end after the $n$'th control operation. Below we will 
identify $t_n = t$ and $t_{n-1} = t-dt$. 

(2) We now turn to the energetic description. As in standard stochastic thermodynamics, we neglect the energetics 
associated to the memory, that is we set $H_U \sim 1_U$ for all units. This implies that we can replace our 
stochastic energy $E_{SU(\bb n)}(t,r_n)$ by $E_S(t,r_n)$. Then, the stochastic energy at the beginning of the interval 
is simply $H(r_{n-1},\lambda_{t-dt})$ and at the end it reads $H(r_n,\lambda_t)$, which is identical to the 
definition used in classical stochastic thermodynamics. Furthermore, in absence of control, we obtain from 
Eq.~(\ref{eq W standard}) 
\begin{align}
 W^{(n)}(r_{n-1})      &=  \sum_s [H(s,\lambda_t) - H(s,\lambda_{t-dt})] p_s(t^+_{n-1},r_{n-1})   \nonumber  \\
                       &=  H(r_{n-1},\lambda_t) - H(r_{n-1},\lambda_{t-dt}),
\end{align}
which is identical to Eq.~(\ref{eq W classical stoch thermo}).\footnote{We remark that there is a certain degree of 
freedom involved in the evaluation of the integral in Eq.~(\ref{eq W standard}). However, this degree of freedom is also 
there in the identification~(\ref{eq W classical stoch thermo}) and~(\ref{eq Q classical stoch thermo}) and it is only 
important to stick consistently to one choice. } Furthermore, the work during the control step, Eq.~(\ref{eq W ctrl}), 
is zero because the marginal state of the system does not change, see also Sec.~\ref{sec projective measurements}. 
Thus, we conclude that the definition of the total work $W^{(n]}(r_{n-1})$ during one full 
interval is identical to the definition used in classical stochastic thermodynamics. 
It remains to look at the change of heat during one full interval $Q_S^{(n]}(r_n,r_{n-1})$. First of all, from 
Eq.~(\ref{eq Q standard}) the heat exchanged during the interval without control becomes 
\begin{equation}
 \begin{split}
  & Q_S^{(n)}(r_{n-1}) = \\
  & \sum_s H(s,\lambda_t) p_s(t_n^-,r_{n-1}) - H(r_{n-1},\lambda_t),
 \end{split}
\end{equation}
which is different from the definition~(\ref{eq Q classical stoch thermo}). However, it is now also important to take 
into account the heat exchanged during the control step, Eq.~(\ref{eq Q ctrl}), in which we update our knowlegde about 
possible system changes. It is simple to see that this quantity reduces to 
\begin{equation}
 \begin{split}
  & Q_S^\text{ctrl}(r_n,r_{n-1}) =  \\
  & H_S(r_n,\lambda_t) - \sum_s H(s,\lambda_t) p_s(t_n^-,r_{n-1}),
 \end{split}
\end{equation}
such that $Q_S^{(n]}(r_n,r_{n-1}) = Q_S^\text{ctrl}(r_n,r_{n-1}) + Q_S^{(n)}(r_{n-1})$ is identical to 
the standard definition in classical stochastic thermodynamics. To conclude, our definitions for stochastic internal 
energy, work and heat are identical to the ones used in classical stochastic thermodynamics. 

(3) We now take a look at the entropic balance. The change in stochastic entropy~(\ref{eq entropy}) over a full interval 
becomes 
\begin{equation}
 \Delta S_{SU(\bb n)}^{(n]}(r_n,r_{n-1}) = -\ln p(r_n|r_{n-1}),
\end{equation}
where we used that the system and units are after each measurement in a pure state and their entropy vanishes. 
Furthermore, we used that the system dynamics are Markovian and hence, $p(r_n|\bb r_{n-1}) = p(r_n|r_{n-1})$. The 
stochastic entropy production~(\ref{eq entropy production}) over one interval then becomes 
\begin{equation}
 \Sigma^{(n]}(r_n,r_{n-1}) = -\ln p(r_n|r_{n-1}) - \beta Q_S^{(n]}(r_n,r_{n-1}),
\end{equation}
which can have either sign. As deduced in Sec.~\ref{sec general stoch thermo}, it is positive after averaging 
over $p(r_n|r_{n-1})$: 
\begin{equation}
 \begin{split}\label{eq ent prod stoch thermo modified 2}
  \Sigma^{(n]}(r_{n-1})  &=  \sum_{r_n} p(r_n|r_{n-1}) \Sigma^{(n]}(r_n,r_{n-1})   \\
                         &=  S_\text{Sh}[p(r_n|r_{n-1})] - \beta Q_S^{(n]}(r_{n-1}) \ge 0.
 \end{split}
\end{equation}
Notice that this second law is identical to the conventional one of stochastic thermodynamics if we apply 
Eq.~(\ref{eq ent prod stoch thermo}) to an initially pure state $p_s(t-dt) = \delta_{s,r_{n-1}}$, which implies 
$\Delta S_\text{ST}(t) = S_\text{Sh}[p(r_n|r_{n-1})]$ and $Q_\text{ST}(t) = Q_S^{(n]}(r_{n-1})$. Unfortunately, 
although $S_\text{Sh}[p(r_n|r_{n-1})]$ is infinitesimal small, it is of order $\C O(dt^\nu)$ with $\nu<1$. Therefore, 
the \emph{rate} of entropy production diverges: 
\begin{equation}
 \lim_{dt\rightarrow0} \frac{\Sigma^{(n]}(r_{n-1})}{dt} = \infty.
\end{equation}
Although seldomly stated~\cite{SpohnJMP1978}, this is related to the fact that the Shannon entropy $S_\text{Sh}[p_s(t)]$ 
is not differentiable when the kernel of $p_s(t)$ changes. Furthermore, by averaging 
Eq.~(\ref{eq ent prod stoch thermo modified 2}) also over $p(r_{n-1})$, we obtain 
\begin{equation}\label{eq ent prod stoch thermo modified}
 \Sigma^{(n]} = S_\text{Sh}(r_n|r_{n-1}) - \beta Q_S^{(n]} \ge 0.
\end{equation}
Here, $S_\text{Sh}(r_n|r_{n-1}) = \sum_{r_{n-1}} p(r_{n-1})S_\text{Sh}[p(r_n|r_{n-1})]$ denotes the conditional 
Shannon entropy. This second law is different from the conventional one~(\ref{eq ent prod stoch thermo}). Instead 
of containing the change in Shannon entropy of the system state, it contains the conditional Shannon entropy, which is 
nothing else than the entropy rate of the stochastic process~\cite{CoverThomasBook1991}. Of course, if we devide 
Eq.~(\ref{eq ent prod stoch thermo modified}) by $dt$, it still diverges. Furthermore, the difference in the two 
entropy productions is precisely given by $\Sigma^{(n]} - \Sigma_\text{ST}(t) = S_\text{Sh}(r_{n-1}|r_n)$. Here, 
the `backward' conditional entropy $S_\text{Sh}(r_{n-1}|r_n)  = \sum_{r_n} p(r_n)S_\text{Sh}[p(r_{n-1}|r_n)]$ 
can be computed via Bayes' rule: $p(r_{n-1}|r_n) = p(r_n|r_{n-1}) p(r_{n-1})/p(r_n)$. 

We emphasize that our novel second law~(\ref{eq ent prod stoch thermo modified}) has a transparent physical 
interpretation. It consists of the entropic change in the bath quantified by the Clausius-like term $-\beta Q_S^{(n]}$ 
plus the change in entropy in our memory for the measurement outcomes. As we measure perfectly \emph{and} continuously, 
the \emph{rate} of information generation in the memory is infinite (in reality, every sampling rate is 
finite and no divergence arises). Therefore, even in equilibrium where $Q_S^{(n]} = 0$, we will have a positive 
entropy production $\Sigma^{(n]} > 0$ due to the fact that we measure the system and continuously generate information. 
In stochastic thermodynamics, one instead finds $\Sigma_\text{ST} = 0$ at equilibrium. 
The discrepancy of the two second laws is rooted in the fact that standard stochastic thermodynamics keeps the 
observer out of the contruction. This works well if one only perfectly monitors a classical system, but if one 
starts to apply feedback control one needs to \emph{modify} the theory~\cite{ParrondoHorowitzSagawaNatPhys2015}. 
By following the credo ``information is physical''~\cite{LandauerPhysTod1991} and by treating the measurement and the 
system on an equal footing, no modification is necessary in our framework. We remark that our novel second 
law~(\ref{eq ent prod stoch thermo modified}) was very recently already experimentally 
confirmed~\cite{RibezziCrivellariRetortNP2019}, see also the discussion in Ref.~\cite{StrasbergWinterArXiv2019}. 

(4) In addition, we can recover the conventional second law of stochastic thermodynamics, if we redefine entropy. 
Namely, if we replace our definition of entropy by the conventional one~(\ref{eq E S stoch thermo}), the stochastic 
entropy production becomes in our notation 
\begin{equation}
 -\ln p(r_n) + \ln p(r_{n-1}) - \beta[H(r_n,\lambda_t) - H(r_{n-1},\lambda_t)].
\end{equation}
If we average over $p(\bb r_n)$ and use that the measured probabilities are identical to the probabilities of the 
system, $p(r_n=s) = p_s(t)$ and $p(r_{n-1}=s) = p_s(t-dt)$, we obtain 
\begin{equation}
 \begin{split}
   & S_\text{Sh}[p_s(t)] - S_\text{Sh}[p_s(t-dt)]   \\
   & - \beta \sum_s H(s,\lambda_t) [p_s(t)-p_s(t-dt)].
 \end{split}
\end{equation}
This is identical to Eq.~(\ref{eq ent prod stoch thermo}). 

\subsection{Getting rid of the units in the thermodynamic description}
\label{sec getting rid of units}

We used the external stream of units to guide our thermodynamic analysis along the framework of repeated interactions. 
In many important realistic situations it is also clear how to model the units physically. This is, for 
instance, the case for the micromaser, the experimental setup studied in Sec.~\ref{sec example} or for certain 
mesoscopic devices where tunneling electrons and Cooper pairs could be identified as 
units~\cite{RodriguesImbersArmourPRL2007, WestigEtAlPRL2017}. Therefore, the 
framework of repeated interactions allows us to treat a larger class of physically relevant scenarios. 

Nevertheless, there are also scenarios where the exact microscopic nature of the units is not known or hard to model. 
Furthermore, as also the process tensor relies only on specifying CP maps $\C A(r_n|\bb r_{n-1})$ acting on the 
\emph{system}, it is worth to ask whether we can get rid of the sometimes rather artifical units in the thermodynamic 
description. Energetically, we have already seen that simply setting $H_{U(n)} \sim 1_U$ for all $n$ cancels out all 
unit contributions from the first law. To get rid of the units from the entropic considerations, we will need to 
restrict ourselves to \emph{efficient} control operations~\cite{WisemanMilburnBook2010, JacobsBook2014}. Efficient 
control operations are defined by the requirement that they can be written as 
\begin{equation}
 \tilde\rho_S(r) = \C A(r)\rho_S = A(r)\rho_S A(r)^\dagger
\end{equation}
as opposed to the more general form~(\ref{eq CP map}). They have the specific property that any initially pure state 
$\rho_S$ gets mapped to a pure state again. Mathematically, every efficient control operation can be modeled by 
an initially pure unit state $\rho_U = |\psi\rl\psi|_U$, which interacts unitarily via $V$ with the 
system and is finally projectively measured using $P(r) = |r\rangle\langle r|$. This implies 
\begin{equation}
 \tilde\rho_S(r) = \C A(r)\rho_S = \big\langle r\big|\C V[\rho_S\otimes|\psi\rl\psi|_U]\big|r\big\rangle_U.
\end{equation}
This construction extends to multiple operations conditioned on previous results $\bb r_{n-1}$ in the obvious way. 

To see that the units also do not enter the entropic balance in this case, notice that 
the unit state is pure and decorrelated from the system after every operation. This follows from the fact that 
we perform a rank 1 projective measurement on the units after each control operation. The joint state of the system and 
all units after obtaining the sequence of outcomes $\bb r_n$ is simply 
$\rho_{SU(\bb n)}(t,\bb r_n) = \rho_{S}(t,\bb r_n)\otimes|\bb r_n\rangle\langle\bb r_n|_{U(\bb n)}$ 
with 
$|\bb r_n\rangle\langle\bb r_n|_{U(\bb n)} \equiv |r_n\rangle\langle r_n|_{U(n)}\otimes\dots\otimes|r_1\rangle\langle r_1|_{U(1)}$. The joint entropy for this state becomes 
$S_\text{vN}[\rho_{SU(\bb n)}(t,\bb r_n)] = S_\text{vN}[\rho_S(t,\bb r_n)]$. 
Also before the interaction at time $t_n$, we have 
\begin{equation}
 S_\text{vN}[\rho_{SU(\bb n)}(t_n^-,\bb r_n)] = S_\text{vN}[\rho_S(t_n^-,\bb r_n)],
\end{equation}
where we used that the initial unit state is pure and hence, always decorrelated from the system. We note that the 
ensemble averaged system unit state $\sum_{\bb r_n} p(\bb r_n) \rho_{SU(\bb n)}(t,\bb r_n)$ is in general classically 
correlated. 

To summarize, in case of energetically neutral units and efficient control operations, the stochastic internal energy 
and entropy can be reduced to 
\begin{align}
 E_S(t,\bb r_n) &=  \mbox{tr}_S\{H_S(\lambda_t,\bb r_n) \rho_S(t,\bb r_n)\},    \\
 S_S(t,\bb r_n) &=  -\ln p(\bb r_n) + S_\text{vN}[\rho_S(t,\bb r_n)].   \label{eq S no units}
\end{align}
Note, however, that we are still using the external units to model the control operations dynamically. 
We will discuss in Sec.~\ref{sec final remarks} how far it is possible to get completely rid of the units.

\subsection{Quantum stochastic thermodynamics without theory input}
\label{sec quantum stoch thermo without theory}

To set up our framework of quantum stochastic thermodynamics, we needed to be able to know the 
work~(\ref{eq W standard}) and heat~(\ref{eq Q standard}) exchanged with the bath in between two control operations. 
Those are path dependent quantities [i.e., they are not determined alone by the state at the boundary 
$\rho_S(t^\pm_n,\bb r_n)$] and estimating them requires additional theoretical input. Albeit this is necessary to 
recover the average picture in general (see Sec.~\ref{sec standard quantum thermo}), it is instructive to discuss 
cases which do not require any additional theoretical modeling. 

Without changing any of our general conclusions, one way would be to consider only a specific subset of 
control protocol $\lambda_t$. These control protocols consist of a sudden switch of the Hamiltonian after each control 
operation, i.e., the protocol changes instantaneously from $\lambda_{n-1}$ to $\lambda_n$ at time $t_n^+$, and after 
the switch we keep the protocol constant until the next control operation. Note that the 
protocol is still allowed to depend on $\bb r_n$, which we have suppressed for notational convenience. Thus, in short we 
can write that $\lambda_t(\bb r_{n-1}) = \lambda_{n-1}(\bb r_{n-1})$ if $t\in(t_{n-1},t_n]$. Those sets of control 
protocols are characterized by the fact that the work~(\ref{eq W standard}) and heat~(\ref{eq Q standard}) can be 
computed without any knowledge about the system state in between two control operations: 
\begin{align}
 & W^{(n)}_S(\bb r_{n-1}) =	\\
 & \mbox{tr}_S\{[H_S(\lambda_{n-1},\bb r_{n-1}) - H_S(\lambda_{n-2},\bb r_{n-2})]\rho_S(t_{n-1}^+,\bb r_{n-1})\} \nonumber  \\
 & Q^{(n)}_S(\bb r_{n-1}) =	\\
 & \mbox{tr}_S\{H_S(\lambda_{n-1},\bb r_{n-1})[\rho_S(t_n^-,\bb r_{n-1}) - \rho_S(t_{n-1}^+,\bb r_{n-1})]\}. \nonumber
\end{align}

Another way to approach this problem is to try to set up an effective thermodynamic description based solely on 
knowledge of the dynamical map $\C E_n$ defined in Eq.~(\ref{eq dynamical map}). Note that the dynamical map can be 
inferred from knowledge of the process tensor. The very problem of this approach comes from the fact that different 
physical situations (with different thermodynamic values for $W^{(n)}_S$ and $Q^{(n)}_S$) can give rise to the same 
dynamical map $\C E_n$. Thus, if we try to pursue the second way, we will not be able to recover the results from 
Secs.~\ref{sec standard quantum thermo} and~\ref{sec repeated interactions ensemble level} in general. Nevertheless, 
the author believes that it could be worthwile to pursue this direction because the thermodynamic description of 
dynamical maps was already investigated before~\cite{AndersGiovannettiNJP2013, BinderEtAlPRE2015, 
ManzanoHorowitzParrondoPRE2015, BarraLledoPRE2017}. Especially, for dynamical maps which have additional properties, 
such as being Gibbs state-preserving, the present framework could be fruitfully combined with the resource 
theory approach to quantum thermodynamics~\cite{GooldEtAlJPA2016, LostaglioArXiv2018}. 

\section{Final remarks and outlook }
\label{sec final}

\subsection{Final remarks}
\label{sec final remarks}

We have presented a theoretical framework, which is able to cope with arbitrary quantum operations and arbitrary 
`unravelings' of them. It uses very natural definitions of internal energy~(\ref{eq internal energy}) and 
entropy~(\ref{eq entropy}), but in its most general form it can appear quite heavy. Especially, the framework of 
repeated interactions added another layer of complexity and it is worthwhile to ask whether we can get completely 
rid of it. For efficient control operations we have seen already in Sec.~\ref{sec getting rid of units} that the 
units do not enter the laws of thermodynamics anymore, albeit they still played a role dynamically. This was important 
in order to arrive at an unambiguous interpretation of heat and work during the control step. 
Let us look at an arbitrary efficient operation $\tilde\rho_S(r) = A(r)\rho_S A(r)^\dagger$ again. It is tempting to 
use the polar decomposition theorem $A(r) = U(r)P(r)$, where $U(r)$ is a unitary matrix and $P(r)$ a positive matrix, 
to define work and heat exchanges. One idea could be to associate changes in the energy caused by $P(r)$ [$U(r)$] as 
heat (work). Unfortunately, one then arrives at the conclusion that a projective measurement is on average a heat 
and not a work source and we have debated this problem already in Sec.~\ref{sec projective measurements}. Moreover, 
there is also a `reverse' polar decomposition theorem $A(r) = P'(r)U(r)$, where $P'(r)\neq P(r)$ in general. This 
would then give rise to a \emph{different} splitting into heat and work for the \emph{same} control operation. This is 
even true in the case $P'(r) = P(r)$ because in the reverse decomposition the positive matrix acts after the unitary. 
By using Stinespring's dilation theorem we have circumvented this difficulty in the repeated interaction framework. 
In this picture the unitary $V$ must always act first to correlate the system and the unit before it is followed by a 
measurement of the unit. This fixes the ambiguity of assigning heat and work, which can be conveniently computed by using 
the control operations only, see Eqs.~(\ref{eq W ctrl S only}) and~(\ref{eq Q ctrl S only}). Thus, for efficient control 
operations with energetically neutral units the explicit modeling of the units is no longer necessary.

Another subtle point concerns the definition of an `entropy production' via a time-reversed process. 
We have here decided to find a meaningful definition of heat and entropy at the first place and we have then checked 
that the entropy production $\Sigma = \Delta S_S - \beta Q$ as known from phenomenological nonequilibrium 
thermodynamics is positive on average. Remarkably, within the framework of classical stochastic thermodynamics there is 
an equivalent alternative approach by defining the stochastic entropy production as 
\begin{equation}\label{eq ent prod time reversed}
 \tilde\Sigma(\bb r_n) \equiv \ln\frac{p(\bb r_n)}{p^\dagger(\bb r_n^\dagger)}.
\end{equation}
Here, $p^\dagger(\bb r_n^\dagger)$ is the probability to observe the time-revered trajectory in a suitably chosen 
time-reversed experiment~\cite{SeifertRPP2012, VandenBroeckEspositoPhysA2015}. This stochastic entropy production 
fulfills a fluctuation theorem and a second law and it is linked to the (breaking of) time-reversal symmetry 
of the underlying microscopic Hamiltonian dynamics~\cite{EvansSearlesAdvPhy2002, EspositoHarbolaMukamelRMP2009, 
CampisiHaenggiTalknerRMP2011, JarzynskiAnnuRevCondMat2011}. 
It is tempting to apply a similar strategy also within our framework by defining a suitable `time-reversed' process 
to construct the `entropy production'~(\ref{eq ent prod time reversed}). Unfortunately, for a general quantum 
operation it is not clear what the corresponding time-reversed process should be. Various proposals have been put 
forward and used in the literature~\cite{CrooksPRA2008, ManikandanJordanQSMF2018, HorowitzPRE2012, 
HorowitzParrondoNJP2013, ManzanoHorowitzParrondoPRE2015, DresselEtAlPRL2017, ElouardEtAlQInf2017, BenoistEtAlCM2018, 
ManikandanElouardJordanPRA2019, ElouardMohammadyBook2018, ManzanoHorowitzParrondoPRX2018} resulting in \emph{multiple} 
possible second laws for the \emph{same physical situation}. It is 
an advantage of the present framework that we are able to derive a second law without taking the detour of 
defining a time-reversed process, which -- at least at the moment -- seems to entail an unwanted amount of ambiguity. 

As a final `final remark' we comment on the possibility to extend the present framework beyond the case of a single 
heat bath. In fact, this is even an open problem in classical stochastic thermodynamics from an experimental point of 
view: as soon as multiple heat baths induce transitions between the same system states, a local measurement of the 
system only will not reveal which bath has triggered the transition. Classically, a way out of this dilemma is to 
experimentally ensure that the transition between each pair of states is only caused by a single bath, for 
instance by geometrically separating the system into subsystems, where each subsystem interacts only with one bath. 
This is indeed what happens in transport experiments through quantum dots~\cite{UtsumiEtAlPRB2010, KungEtAlPRX2012}. 
Quantum mechanically, this separation is more difficult to achieve. At least within the standard approach based on a 
Born-Markov-secular approximation~\cite{SpohnLebowitzAdvChemPhys1979, AlickiJPA1979, LindbladBook1983, 
KosloffEntropy2013}, the system jumps between energy eigenstates of the composite system which are in general entangled. 
On the other hand, it was recently also argued that a `local' approach to the dynamics (where each dissipator in a 
quantum system acts only on a specific subsystem) is feasible from a thermodynamic point of view~\cite{BarraSciRep2015, 
TrushechkinVolovichEPL2016, HoferEtAlNJP2017}. If that is the case, it should be in principle possible to apply our 
framework to a situation with multiple baths in some limit. As the proper extension of quantum thermodynamics to the 
presence of multiple heat baths can already bear surprising difficulties at the average 
level~\cite{MitchisonPlenioNJP2018}, these investigations are left for the future. 

\subsection{Outlook}
\label{sec outlook}

In this last section we outline three promising future applications that allow us to answer in a general 
and rigorous way open problems in quantum thermodynamics. 

\subsubsection{Quantum coherence and Leggett-Garg inequalities}

One primary task of quantum thermodynamics is to unravel how quantum features (such as coherence or entanglement) 
influence the performance of quantum heat engines and other devices. An introduction to this topic was recently 
provided in Ref.~\cite{LevyGelbwaserKlimovskyBook2018}. While several interesting results have been found (showing that 
quantum effects can be both, beneficial and detrimental), one always has to be cautious when \emph{comparing} them with 
classical systems. In fact, it is far from obvious to which extend quantum and classical models can be compared and what 
are genuine quantum features. For instance, the mere presence of coherences (i.e., off-diagonal elements of the density 
matrix in the energy eigenbasis) is not sufficient to conclude that the heat engine operates in the 
`quantum regime'~\cite{OnamGonzalezEtAlPRE2019}. As we will show now, our framework allows us to rigorously answer 
whether a \emph{given} heat engine uses quantum coherence. Moreover, this is closely related to the violations of 
Leggett-Garg inequalities~\cite{EmaryLambertNoriRPP2014}. 

Our analysis is based on recent progress to understand genuine quantum effects in Markovian systems interrupted by 
projective measurements at a set of discrete times~\cite{SmirneEtAlQST2018, StrasbergDiazPRA2019, 
MilzEgloffEtAlArXiv2019}. In a nutshell, the authors of Ref.~\cite{SmirneEtAlQST2018} have proven that the results 
$\bb r_n$ obtained from the projective measurements in an arbitrary non-degenerate basis $\{|r_n\rangle\}$ cannot be
generated by a classical stochastic process if and only if the Markovian dynamics are 
``coherence-generating-and-detecting'' for an initially diagonal state in the measurement basis. The notion 
coherence-generating-and-detecting is defined by using the dephasing operator $\C D = \sum_{r_n} \C P(r_n)$, where 
$\C P(r_n)$ denotes the projection superoperator with respect to $|r_n\rl r_n|$, and by demanding that there exists times 
$t, \tau \ge 0$ such that 
\begin{equation}
 \C D\circ \C E(t)\circ \C D\circ \C E(\tau)\circ \C D \neq \C D\circ \C E(t+\tau)\circ\C D.
\end{equation}
where $\C E(t)$ denotes the dynamical map of the system in between the control operations (here assumed to be 
time-homogeneous for simplicity) and $\circ$ the composition of two maps. An extension to inhomogeneous maps and more 
general (i.e., non-Markovian) dynamics can be found in Refs.~\cite{StrasbergDiazPRA2019, MilzEgloffEtAlArXiv2019}. 

This framework fits perfectly into our language as we can deal with projective measurements at discrete 
times as well as dephasing operations. To give a simple and intuitive example how this framework could be used to 
detect quantum signatures in thermodynamics, we consider the quantum Otto cycle, which was recently also experimentally 
realized~\cite{RossnagelEtAlScience2016}. The Otto cycle is a four-step process 
$A\rightarrow B\rightarrow C\rightarrow D$ (see, e.g., Fig.~2 in Ref.~\cite{LevyGelbwaserKlimovskyBook2018}). 
In $A\rightarrow B$ the system undergoes isolated (unitary) Hamiltonian evolution, where the system Hamiltonian changes 
from $H_S(1)$ to $H_S(2)$. In $B\rightarrow C$ the system is coupled to a cold bath at temperature $T_C$ and undergoes 
pure relaxation dynamics, which we here assume to be modeled by a Lindblad master equation as often done. In 
$C\rightarrow D$ the system is again isolated and its Hamiltonian is changed from $H_S(2)$ back to $H_S(1)$ again. 
Finally, in $D\rightarrow A$ the system is coupled to a hot bath at temperature $T_H$ and undergoes again pure 
relaxation assumed to be described by a Lindblad master equation. 
After the unitary strokes at point $B$ and $D$ the system density matrix contains coherences in general. If the 
cycle is performed in finite time such that the heat baths do not fully erase the coherences, then it is possible 
that coherences are still present at point $A$ and $C$ and it becomes an interesing question whether they change the 
thermodynamic performance. 

To unambiguously answer this question, one could perform a dephasing operation $\C D$ in the energy eigenbasis at 
any of the four points. If this changes the work output or the thermodynamic efficiency\footnote{Note that we have 
not specified here how to actually infer the work output or efficiency. This could be done purely theoretically or 
purely experimentally, for instance, by doing quantum state tomography at the four points $A, B, C$ and $D$ after waiting long enough such that the system operates at steady state (actually, state tomography at two points suffices 
if we are able to accurately compute the effect of the unitary strokes). }, then the machine shows quantum 
effects. The dephasing operation $\C D$ is easily implemented, for instance, by performing a projective measurement 
of the energy without recording its outcome (see Sec.~\ref{sec projective measurements}). Importantly, the energetic 
cost of this control operation is zero (provided that the unit is energetically neutral) and therefore, it does not 
inject or extract any work into the engine. Hence, while the dephasing operation has an entropic cost, this does not 
play any role to compute the work and heat flows in the Otto cycle, which are essential to compute its performance. 

Conversely, by the theorem derived in Ref.~\cite{SmirneEtAlQST2018, StrasbergDiazPRA2019, MilzEgloffEtAlArXiv2019}, 
we know that coherences can only influence the 
dynamics, if the statistics associated with projective energy measurements at any subset of the four points in the Otto 
cycle shows non-classical signatures by not obeying the Kolmogorov consistency condition. For instance, if the dephasing 
operation at point $B$ has an influence on the thermodynamic performance, then also 
\begin{equation}\label{eq inequality Kolmogorov consistency}
 \sum_{r_B} p(r_A,r_B,r_C) \neq p(r_A,r_C).
\end{equation}
Here, we have denoted the outcome of the projective measurement at point $A$ by $r_A$ (and analogously for the other 
points) in spirit of our previous notation. Thus, instead of looking at the effect of the dephasing operation, we could 
also alternatively use the process tensor formalism to infer the statistics of the projective measurements directly. 

Remarkably, Eq.~(\ref{eq inequality Kolmogorov consistency}) is a necessary prerequisite to violate the Leggett-Garg 
inequality~\cite{EmaryLambertNoriRPP2014}. Thus, by probing the multitime correlations of a quantum 
stochastic process we can also learn something about its thermodynamic behaviour and unravel the regime where it has no 
analogous classical stochastic thermodynamic process. First results in this direction have been already obtained in 
Refs.~\cite{LostaglioPRL2018, MillerAndersEntropy2018}. In addition, there are also entropic Leggett-Garg 
inequalities~\cite{MorikoshiPRA2006, UshaDeviEtAlPRA2013, EmaryLambertNoriRPP2014}, which relate 
Eq.~(\ref{eq inequality Kolmogorov consistency}) to the entropy of the measurement result $H(r_C,r_B,r_A)$. As this 
quantity plays a crucial role in our second law, it would be interesting to investigate whether a Maxwell demon 
can extract more or less work from a system and measurement process able to violate the entropic Leggett-Garg 
inequalities. 

\subsubsection{Entanglement}

Closely related to the previous analysis is the question how far entanglement can boost the performance of a heat 
engine. There has been much theoretical progress on understanding the role of entanglement for work extraction 
(see, e.g., Refs.~\cite{OppenheimEtAlPRL2002, ZurekPRA2003, AlickiFannesPRE2013, HovhannisyanPRL2013, 
PerarnauLlobetEtAlPRX2015, ManzanoPlastinaZambriniPRL2018}), mostly, however, for extracting work in idealized 
protocols. To the best of the author's knowlegde, a realizable and continuously working heat engine using quantum 
entanglement has not yet been presented. In contrast, classical correlations are known to be indispensible for 
autonomous multipartite heat engines such as thermoelectric devices~\cite{SanchezBuettikerPRB2011, 
StrasbergEtAlPRL2013, HartmannEtAlPRL2015, ThierschmannEtAlNatNanotech2015, KoskiEtAlPRL2015}. 

To test whether a thermodynamic process is influcenced by entanglement, consider a bipartite system $AB$ living in the 
Hilbert space $\C H_A\otimes\C H_B$ as the working fluid. One could then follow a similar strategy as above, but this 
time -- instead of applying a dephasing operation -- one would apply an `entanglement-breaking' operation $\C B$, which 
keeps classical correlations. If the reduced state of system $A$ is given by 
$\rho_A = \sum_i \lambda_i |i\rl i|_A$, then the control operation 
\begin{equation}
 \C B\rho_{AB} = \sum_i |i\rl i|_A\rho_{AB}|i\rl i|_A
\end{equation}
would destroy any entanglement but keep all classical correlations. Monitoring the response of a multipartite system 
to such a control operations then allows the experimenter to infer how far quantum correlations play a role 
thermodynamically. As above, this procedure exemplifies how useful generalized control operation are, not only to control 
a thermodynamic process but also to unravel specific properties of it. 

\subsubsection{Non-Markovian signatures in heat engines}

The last part of this outlook probably requires the largest research effort, but it seems to be necessary in order 
to obtain a complete framework of stochastic thermodynamics for small quantum systems. Indeed, for sufficiently 
low temperatures and sufficiently small time-scales (i.e., where the standard Born-Markov secular master equation 
fails) it is expected that generic open quantum systems behave non-Markovian. Furthermore, even at room temperature 
there is evidence that non-Markovianity can drastically effect bio-chemical processes such as 
photosynthesis~\cite{LambertEtAlNatPhys2013, HuelgaPlenioCP2013} and there is evidence that non-Markovian 
effects can also boost the performance of heat engines~\cite{BylickaEtAlSciRep2016, StrasbergEtAlNJP2016, 
WertnikEtALJCP2018}. Despite the fact that there are several ways to rigorously quantify non-Markovianity in 
open quantum systems~\cite{RivasHuelgaPlenioRPP2014, BreuerEtAlRMP2016}, establishing a rigorous connection between 
thermodynamics and non-Markovianity has proven to be challenging so far~\cite{StrasbergEspositoPRE2019}. 

Notice that the present framework crucially hinges on the assumptions of a Markovian system evolution. However, 
it is not unlikely that it is possible to overcome the assumptions from Sec.~\ref{sec preliminary considerations}. 
One route could be to enlarge the system space by incorporating explicitly the most dominant degrees of freedom of the 
environment into the dynamics -- a strategy which was directly or indirectly proposed in 
Refs.~\cite{StrasbergEtAlNJP2016, KatzKosloffEnt2016, NewmanMintertNazirPRE2017, StrasbergEspositoPRE2017, 
PerarnauLlobetEtAlPRL2018, SchallerEtAlPRB2018, StrasbergEtAlPRB2018, RestrepoEtALNJP2018, WertnikEtALJCP2018}. 
Preliminary results also show that this is not even necessary if we do not consider real-time feedback 
control~\cite{StrasbergTBP}. 

To outline how it would be possible to rigorously detect non-Markovian effects in quantum thermodynamics, 
we make use of the notion of a `causal break'. This notion was recently introduced in Ref.~\cite{PollockEtAlPRL2018} 
to give a general and rigorous definition of non-Markovianity based on the process tensor, which generalizes previous 
attempts~\cite{RivasHuelgaPlenioRPP2014, BreuerEtAlRMP2016}. The basic idea is to apply a control operation to the 
system, which re-prepares it in a state 
independent of all past events. Any dependence of future events on past events then reveals non-Markovian effects. 

To have a particular application in quantum thermodynamics in mind, imagine a steadily working heat engine. The details 
of the machine -- i.e., whether it uses multiple heat baths or feedback control as a resource and whether it acts as a  refrigerator or thermoelectric device -- do not matter for the present consideration. Furthermore, let us denote the 
steady state of the machine by $\bar\rho_S$. Now, as a causal break we apply a control operation which replaces 
the current state of the system by the steady state $\bar\rho_S$. This is always possible: we could, for instance, 
projectively measure the state of the system and then prepare the state $\bar\rho_S$. Since $\bar\rho_S$ will be in 
general mixed, this preparation procedure will be probabilistic (i.e., described by multiple Kraus operators and not 
a single one). The crux is now to apply this control operation when the machine has already reached steady state, 
i.e., we effectively replace $\bar\rho_S$ by $\bar\rho_S$ on average. When the system behaves Markovian, the future 
statistics of all measurements will not depend on this re-preparation procedure, but if the system behaves 
non-Markovian, there will be observable consequences as our control operation has destroyed all time-correlations of 
the system with the past. To see whether such a causal break has an influence on the \emph{thermo}dynamics (which does 
not need to be the case even when the overall dynamics are non-Markovian), one could measure, e.g., the work output 
of the device or its efficiency. Since the system was assumed to operate at steady state, any change in its 
thermodynamic behaviour after the causal break described above unambigously reveals non-Markovian effects. 

Thus, to summarize, we are only beginning to explore quantum effects in thermodynamics. To access those quantum effects 
in a lab, it is important to be able to apply various control operations to the system. The present paper provides the 
toolbox to describe these control operations thermodynamically even along a single stochastic trajectory. 

\section*{Acknowledgements}

It is my pleasure to acknowledge useful discussions with and comments from Cyril Elouard, Massimiliano Esposito and 
Kavan Modi. This research was financially supported by the European Research Council project NanoThermo (ERC-2015-CoG 
Agreement No. 681456) and the DFG (project STR 1505/2-1). Also, it was supported by the Spanish MINECO FIS2016-80681-P 
(AEI-FEDER, UE) and in part by the National Science Foundation under Grant No. NSF PHY17-48958.


\bibliography{/home/philipp/Documents/references/books,/home/philipp/Documents/references/open_systems,/home/philipp/Documents/references/thermo,/home/philipp/Documents/references/info_thermo,/home/philipp/Documents/references/general_QM,/home/philipp/Documents/references/math_phys,/home/philipp/Documents/references/general_refs}

\end{document}